\documentclass[11pt,a4paper]{article}
\usepackage{mathrsfs}
\usepackage{amsmath,amssymb,amsthm}
\usepackage{cite}
\usepackage{color}
\usepackage{hyperref}

\newcommand{\balpha}{\boldsymbol{\alpha}}

\newcommand{\bfx}{\mathbf{x}}
\newcommand{\bfy}{\mathbf{y}}

\newtheorem{lem} {Lemma}[section]
\newtheorem{prop}{Proposition}[section]
\theoremstyle{definition}
\newtheorem{defn} {Definition}[section]
\theoremstyle{remark}

\numberwithin{equation}{section}

\newcommand{\bse}{\begin{subequations}}
\newcommand{\ese}{\end{subequations}}

\def \b#1{\bar{#1}}

\def \c#1{\mathcal{#1}}
\def \f#1{\mathfrak{#1}}
\def \ss#1{\hbox{\tiny{[#1]}}}
\begin{document}

\title{The semi-discrete AKNS system: Conservation laws, reductions\\ and continuum limits}
\author{Wei Fu$^{1}$,
\quad
Zhijun Qiao$^{2}$,\quad
Junwei Sun$^{1}$,\quad
Da-jun Zhang$^{1}$\footnote{E-mail address: djzhang@staff.shu.edu.cn}\\
{\small \it $^{1}$Department of Mathematics, Shanghai University, Shanghai 200444, P.R.China}\\
{\small \it $^{2}$Department of Mathematics, The University of Texas-Pan American, Edinburg, Texas 78541, U.S.A.}
}
\date{\today}

\maketitle

\begin{abstract}

In this paper, the semi-discrete Ablowitz-Kaup-Newell-Segur (AKNS) hierarchy is shown in spirit composed by the Ablowitz-Ladik flows under certain combinations.
Furthermore, we derive its explicit Lax pairs and infinitely many conservation laws, which are non-trivial in light of continuum limit.
Reductions of the semi-discrete AKNS hierarchy are investigated to include the semi-discrete Korteweg-de Vries (KdV), the semi-discrete modified KdV, and the
semi-discrete nonlinear Schr\"odinger hierarchies as its special cases.
Finally, under the uniform continuum limit we introduce in the paper,
the above  results of the semi-discrete AKNS hierarchy, including Lax pairs, infinitely many conservation laws
and reductions, recover their counterparts of the continuous AKNS hierarchy.

\end{abstract}


\section{Introduction}\label{sec:intro}

Discrete systems are attracting more and more attentions in the study of nonlinear lattice models, but more complicated and harder than continuous systems, mainly due to nonlocal forms of discrete systems with a lack of Leibniz rule and non-uniqueness of discretization.
A discrete model could be related to different continuous models because it depends on different continuum limits. Conversely, a continuous equation could also have several discrete versions upon the discretization procedure.
Generally speaking, a nonlinear system is called ``integrable" , that means the system can be exactly solved,
or it possesses enough solvable characteristics, such as a Lax pair, enough independent conserved quantities
and symmetries.
These integrable characteristics
should be kept for discretization of continuous integrable systems
so that discretized systems are still integrable.
If a system is not fully discrete (e.g. with a discrete spatial variable $n$ and a continuous temporal variable $t$),
we call it a semi-discrete system.

As one of integrable discretization techniques,
the Miwa transformation
\cite{Miwa-PJA-1982} gave rise to a discretization for the famous Sato theory \cite{DJM-JPSJ-1982-I,DJM-JPSJ-1982-II,
DJM-JPSJ-1983-III,DJM-JPSJ-1983-IV,DJM-JPSJ-1983-V}.
Miwa's transformation added discrete independent variables into the continuous dispersion relation,
which essentially breaks the original spatial and temporal independence into the discretization.
As a consequence, when taking continuum limit to recover a continuous integrable system
from a discrete integrable system obtained through the discrete Sato theory,
one needs to allocate independent variables so that the new independent variables  coincide with the desired continuous dispersion relation.

One can also discretize spectral problems in the Lax pairs by replacing
the derivatives of eigenfunctions with their differences.
For example, the Ablowitz-Ladik (AL) spectral problem \cite{AL75,AL76}
\begin{equation}
\left(
  \begin{array}{c}
     \phi_{1,n+1} \\
     \phi_{2,n+1} \\
  \end{array}
  \right)
  =
\left(
  \begin{array}{cc}
    \lambda & Q_n \\
    R_n & \frac{1}{\lambda}
  \end{array}
\right)
\left(
  \begin{array}{c}
    \phi_{1,n} \\
    \phi_{2,n}
  \end{array}
\right)
\end{equation}
is a discretized version of the well known Ablowitz-Kaup-Newell-Segur (AKNS) spectral problem \cite{AKNS73,AKNS74}
\begin{equation}
\left(
  \begin{array}{c}
     \phi_1 \\
     \phi_2 \\
  \end{array}
  \right)_x
  =
\left(
  \begin{array}{cc}
    \eta & q \\
    r & -\eta
  \end{array}
\right)
\left(
  \begin{array}{c}
    \phi_1 \\
    \phi_2
  \end{array}
\right).
\label{akns-sp}
\end{equation}
In this discretization approach, the spatial and temporal independence is kept, but we still lose the correspondence
between discrete equations and their continuous counterparts.
Obviously, the forward difference $\phi_{j,n+1}-\phi_{j,n}$ is not the unique discretization representative of $\phi_{j,x}$.
Therefore, we usually need to combine the equations in the AL hierarchy so as to obtain the semi-discrete AKNS (sdAKNS) hierarchy.
In a recent paper \cite{ZC10b}, the combinatorial relation between the AL hierarchy and the AKNS hierarchy
is  proved through the continuum limit and the algebra on infinitely many symmetries.

In this paper we will revisit the sdAKNS hierarchy.
It is well known that the AKNS spectral problem \eqref{akns-sp} provides integrable backgrounds for
several nonlinear systems of physical significance \cite{AKNS73}.
The Korteweg-de Vries (KdV) equation, the modified  Korteweg-de Vries (mKdV) equation,
the sine-Gordon   equation, and the nonlinear Schr\"odinger (NLS) equation can be derived as reductions of the AKNS hierarchy.
The reductions can also pass the integrable characteristics of the AKNS hierarchy to the reduced systems,
such as the infinitely many conservation laws \cite{wsk-75}, infinitely many symmetries \cite{LZ86,CZ-JPA-1991}, and so on.
For the AL hierarchy, its integrable characteristics, such as conservation laws, symmetries and Hamiltonian structures,
have well been studied \cite{ZW-JPA-1995,TM-JPSJ-2000,ZC02a,ZC02b,ZNBC06,GHMT08,ZC10a,ZC10b}.

 The main purpose in this paper is to give Lax pairs, infinitely many conservation laws and reductions of the sdAKNS hierarchy,
and investigate their continuum limits, particularly, the conservation laws.
The so-called infinitely many conserved densities derived from the AL spectral problem (cf. \cite{ZC02a})
look like nontrivial but in fact they are trivial in light of continuum limit.
Later on, we will see  that all of these conserved densities
go to the first conserved density of the AKNS hierarchy in the same continuum limit.
It is necessary for one to derive new and nontrivial infinitely many conservation laws for the sdAKNS hierarchy.

The whole paper is organized as follows.
In Sec. 2 we recall some results of the AKNS systems, including hierarchies, Lax pairs, conservation laws, and reductions.
In Sec. 3, first, we re-derive the sdAKNS hierarchy and their Lax pairs
so that they are ready to consider continuum limits.
The Lax pairs are also used to derive conservation laws.
We derive new infinitely many conservation laws, which are not trivial under our continuum limit.
The explicit combinatorial relations between the known conservation laws and the new conservation laws are proved.
As reductions of the sdAKNS hierarchy, we obtain the semi-discrete KdV (sdKdV),
semi-discrete mKdV (sdmKdV), and semi-discrete NLS (sdNLS) hierarchies
together with their recursion operators.
Finally, in Sec. 4 we present a uniform continuum limit, and under the limit the above results of the sdAKNS system recover their counterparts of the continuous AKNS system.

\section{The AKNS system}\label{sec:A}

Let us briefly review the well known continuous AKNS system.
We just display some results of our interest in the paper.
They will be targeted at the continuum limits of the sdAKNS system.

\subsection{The AKNS hierarchy and Lax pairs}\label{sec:A-hie}

The AKNS spectral problem coupled with a time evolution part reads \cite{AKNS73}
\bse\label{A-Lax}
\begin{align}
&\Phi_x=M_{\ss A}\Phi,\quad
M_{\ss A}=
\left(
  \begin{array}{cc}
    \eta & q \\
    r & -\eta
  \end{array}
\right),\quad
u=
\left(
  \begin{array}{c}
    q \\
    r
  \end{array}
\right),\quad
\Phi=
\left(
  \begin{array}{c}
    \phi_1  \\
    \phi_2
  \end{array}
\right),
\label{A-Lax-a}\\
&\Phi_{t_s}=N_{\ss A,s}\Phi,\quad
N_{\ss A,s}=
\left(
  \begin{array}{cc}
    A_{\ss A,s} & B_{\ss A,s} \\
    C_{\ss A,s} & -A_{\ss A,s}
  \end{array}
\right),
\quad s=0,1,2\cdots,\label{A-Lax-b}
\end{align}
\ese
where $\eta$ is a spectral parameter independent of time,
$q=q(x,t)$ and $r=r(x,t)$ are potential functions
and $A_{\ss A,s},B_{\ss A,s}$ and $C_{\ss A,s}$  are polynomials of $\eta$ depending on $q,r$ and their derivatives.
The sub-$s$ is used in the paper to correspond the $s$-th equation in the AKNS hierarchy.
We also note that in the paper the subscripts [A], [K], [M], [N] and [V] are added to denote the notations for the
AKNS, KdV, mKdV, NLS and Volterra systems, respectively.
To obtain the AKNS hierarchy, from the zero curvature equation
\begin{align}\label{A-zcr}
\partial_{t_s}  M_{\ss A}-\partial_x N_{\ss A,s}+[M_{\ss A},N_{\ss A,s}]=0,\quad s=1,2,\cdots,
\end{align}
where $[M_{\ss A},N_{\ss A,s}]=M_{\ss A}N_{\ss A,s}-N_{\ss A,s} M_{\ss A}$,
one can expand
\[
\left(
  \begin{array}{c}
    B_{\ss A,s} \\
    C_{\ss A,s} \\
  \end{array}
\right)
=\sum_{k=1}^s
\left(
  \begin{array}{c}
    b_k \\
    c_k \\
  \end{array}
\right)
(2\eta)^{s-k}
\]
and then rewrite \eqref{A-zcr} as the following,
\bse
\begin{align}
\label{A-hie-1}
& u_{t_s}=K_{\ss A,s}=L_{\ss A}^s
\left(
  \begin{array}{c}
    q \\
    -r \\
  \end{array}
\right),\\
\label{A-BC}
&\left(
  \begin{array}{c}
    B_{\ss A,s} \\
    C_{\ss A,s} \\
  \end{array}
\right)
=-\sigma\sum_{k=1}^sL_{\ss A}^{k-1}
\left(
  \begin{array}{c}
    q \\
    -r \\
  \end{array}
\right)
(2\eta)^{s-k},\\
&A_{\ss A,s}=-\partial_x^{-1}(r,-q)
\left(
  \begin{array}{c}
    B_{\ss A,s} \\
    C_{\ss A,s} \\
  \end{array}
\right)
+\frac{1}{2}(2\eta)^{s},\label{A-A}
\end{align}
\ese
where $L_{\ss A}$ is the recursion operator defined as
\begin{align}\label{A-ro}
L_{\ss A}=-\sigma\partial_x+2\sigma
\left(
  \begin{array}{c}
    q \\
    r \\
  \end{array}
\right)
\partial_x^{-1}
(r,q),\quad
\sigma=
\left(
  \begin{array}{cc}
    -1 & 0 \\
    0 & 1 \\
  \end{array}
\right).
\end{align}
\eqref{A-hie-1} is referred to as the AKNS hierarchy.
Actually, the hierarchy can start from $s=0$ by defining $K_{\ss A,0}=(q,-r)^T$.
Thus, the AKNS hierarchy is expressed as
\begin{align}\label{A-hie}
u_{t_s}=K_{\ss A,s}=L_{\ss A}^s
\left(
  \begin{array}{c}
    q \\
    -r \\
  \end{array}
\right),\quad
s=0,1,\cdots.
\end{align}
The first four flows are
\bse\label{A-Ks}
\begin{align}
&K_{\ss A,0}=
\left(
  \begin{array}{c}
    q \\
    -r \\
  \end{array}
\right),\label{A-K0}\\
&K_{\ss A,1}=
\left(
  \begin{array}{c}
    q_x \\
    r_x \\
  \end{array}
\right),\label{A-K1}\\
&K_{\ss A,2}=
\left(
  \begin{array}{c}
    q_{xx}-2q^2r \\
    -r_{xx}+2qr^2 \\
  \end{array}
\right),\label{A-K2}\\
&K_{\ss A,3}=
\left(
  \begin{array}{c}
    q_{xxx}-6qrq_x \\
    r_{xxx}-6qrr_x \\
  \end{array}
\right).\label{A-K3}
\end{align}
\ese
\eqref{A-Lax} is called the Lax pair of the hierarchy \eqref{A-hie}, where
the first four $N_{\ss A,s}$ read 
\bse\label{A-Ns}
\begin{align}
&N_{\ss A,0}=
\left(
  \begin{array}{cc}
    \frac{1}{2} & 0 \\
    0 & -\frac{1}{2}
  \end{array}
\right),\label{A-N0}\\
&N_{\ss A,1}=
\left(
  \begin{array}{cc}
    \eta & q \\
    r & -\eta
  \end{array}
\right),\label{A-N1}\\
&N_{\ss A,2}=
\left(
  \begin{array}{cc}
    2\eta^2-qr & 2\eta q+q_x \\
    2\eta r-r_x & -2\eta^2+qr
  \end{array}
\right),\label{A-N2}\\
&N_{\ss A,3}=
\left(
  \begin{array}{cc}
    4\eta^3-2\eta qr+qr_x-q_xr & 4\eta^2q+2\eta q_x+q_{xx}-2q^2r \\
    4\eta^2r-2\eta r_x+r_{xx}-2qr^2 & -4\eta^3+2\eta qr-qr_x+q_xr
  \end{array}
\right).\label{A-N3}
\end{align}
\ese

\subsection{Conservation laws}\label{sec:A-cls}

The infinitely many conservation laws of the AKNS hierarchy can be constructed from their Lax pairs (cf. \cite{wsk-75}).
Starting from the AKNS spectral problem \eqref{A-Lax-a}, i.e.
\bse\label{A-sp}
\begin{align}
&\phi_{1,x}=\eta\phi_1+q\phi_2,\label{A-sp-a}\\
&\phi_{2,x}=r\phi_1-\eta\phi_2,\label{A-sp-b}
\end{align}
\ese
and setting $\omega_{\ss A}=\frac{\phi_2}{\phi_1}$,  we are able to obtain the Riccati equation
\begin{align}\label{A-Ric}
2\eta\omega_{\ss A}=-\omega_{\ss A,x}-q\omega_{\ss A}^2+r.
\end{align}
This equation is solved by a series-form
\begin{align}\label{A-ome}
\omega_{\ss A}=\sum_{j=1}^{\infty}\omega_{\ss A}^{(j)}(2\eta)^{-j}
\end{align}
with
\bse\label{A-ome-rec}
\begin{align}
&\omega_{\ss A}^{(1)}=r,\quad \omega_{\ss A}^{(2)}=-r_x,\label{A-ome-rec-a}\\
&\omega_{\ss A}^{(j+1)}=-\omega_{\ss A,x}^{(j)}-q\sum_{k=1}^{j-1}\omega_{\ss A}^{(k)}\omega_{\ss A}^{(j-k)},\quad j=2,3,\cdots.\label{A-ome-rec-b}
\end{align}
\ese
Here we list the first few of $\omega_{\ss A}^{(j)}$:
\bse\label{A-omej}
\begin{align}
&\omega_{\ss A}^{(1)}=r,\label{A-ome0}\\
&\omega_{\ss A}^{(2)}=-r_x,\label{A-ome1}\\
&\omega_{\ss A}^{(3)}=r_{xx}-qr^2,\label{A-ome2}\\
&\omega_{\ss A}^{(4)}=-r_{xxx}+q_xr^2+4qrr_x.\label{A-ome3}
\end{align}
\ese
Next, from the Lax pair \eqref{A-Lax} one can find
\bse\label{A-lnphi}
\begin{align}
&(\ln\phi_2)_x=\eta+q\omega_{\ss A},\label{A-lnphi-a}\\
&(\ln\phi_2)_{t_s}=A_{\ss A,s}+B_{\ss A,s}\omega_{\ss A},\label{A-lnphi-b}
\end{align}
\ese
where $A_{\ss A,s}$ and $B_{\ss A,s}$ are expressed by \eqref{A-A} and \eqref{A-BC}, respectively.
The compatibility condition $(\ln \phi_2)_{x,t_s}=(\ln \phi_2)_{t_s,x}$ leads to a formal conservation law
\begin{align}\label{A-cf}
(q\omega_{\ss A})_{t_s}=(A_{\ss A,s}+B_{\ss A,s}\omega_{\ss A})_x.
\end{align}
To derive  infinitely many conservation laws, one needs to insert \eqref{A-ome} into \eqref{A-cf}
and expand \eqref{A-cls} into a series in terms of $2\eta$.
We note that after expansion one finds
\begin{equation}
A_{\ss A,s}+B_{\ss A,s}\omega_{\ss A}=\frac{1}{2}(2\eta)^{s}+\sum^{\infty}_{j=1}J_{\ss A,s}^{(j)}(2\eta)^{-j},
\label{A-expn-flux}
\end{equation}
where the term $\frac{1}{2}(2\eta)^{s}$ comes from $A_{\ss A,s}|_{u=0}$ and contributes nothing to the conservation laws.
The coefficients of every different  power of $2\eta$ in \eqref{A-cf}  compose the infinitely many conservation laws
\begin{align}\label{A-cls}
\partial_{t_s}\rho_{\ss A}^{(j)}=\partial_xJ_{\ss A,s}^{(j)},\quad j=1,2,\cdots.
\end{align}
We note that the infinitely many conserved densities $\{\rho_{\ss A}^{(j)}\}$ are shared by all the equations in the AKNS hierarchy.
The first three of the conserved densities are
\bse\label{A-rhoj}
\begin{align}
&\rho_{\ss A}^{(1)}=qr,\label{A-rho1}\\
&\rho_{\ss A}^{(2)}=-qr_x,\label{A-rho2}\\
&\rho_{\ss A}^{(3)}=qr_{xx}-q^2r^2.\label{A-rho3}
\end{align}
\ese
The associated fluxes depend on which equation is considered in the hierarchy.
For example, if we consider $u_{t_2}=K_{\ss A,2}$,
we have $A_{\ss A,2}= 2\eta^2-qr $ and $B_{\ss A,2}= 2\eta q+q_x$
and from \eqref{A-expn-flux} the first three fluxes are
\bse\label{A-Jj}
\begin{align}
&J_{\ss A,2}^{(1)}=q_xr-qr_x,\label{A-J1}\\
&J_{\ss A,2}^{(2)}=qr_{xx}-q_xr_x-q^2r^2,\label{A-J2}\\
&J_{\ss A,2}^{(3)}=q_xr_{xx}-qr_{xxx}+4q^2rr_x.\label{A-J3}
\end{align}
\ese

\subsection{Reductions}\label{sec:A-reduc}

It is known that the KdV, mKdV, NLS, and sine-Gordon families are able to be derived from the AKNS hierarchy \eqref{A-hie}
as its reductions.
Here, let us only list the cases of the KdV, mKdV, and NLS systems.

\subsubsection{The KdV family}\label{sec:KdV}

The KdV family is reduced from the odd numbered equations in the AKNS hierarchy by taking $(q,r)=(q,-1)$, i.e.
\begin{align}\label{K-reduc}
\left(
  \begin{array}{c}
    q \\
    -1 \\
  \end{array}
\right)_{t_{2s+1}}
=K_{\ss A,2s+1}|_{r=-1}=(L_{\ss A}|_{r=-1})^{2s}
\left(
  \begin{array}{c}
    q_x \\
    0 \\
  \end{array}
\right)=
\left(
  \begin{array}{c}
    L_{\ss K}^sq_x \\
    0 \\
  \end{array}
\right),\quad
s=0,1,\cdots,
\end{align}
where
\begin{align}\label{K-ro}
L_{\ss K}=\partial_x^2+4q+2q_x\partial_x^{-1}
\end{align}
is the recursion operator of the KdV family.
\eqref{K-reduc} is simplified to the following,
\begin{align}\label{K-hie}
q_{t_{2s+1}}=K_{\ss K,2s+1}=L_{\ss K}^sq_x,\quad s=0,1,\cdots,
\end{align}
which are referred to as the KdV family.
When $s=1$, we obtain the KdV equation
\begin{align}\label{K-eq}
q_{t_3}=K_{\ss K,3}=q_{xxx}+6qq_x.
\end{align}
If we eliminate $\phi_1$ from \eqref{A-Lax} with $s=1$ and the constraint $(q,r)=(q,-1)$,
we obtain the following scalar form Lax pair for the KdV equation:
\bse\label{K-Lax}
\begin{align}
&\phi_{xx}=(\eta^2-q)\phi,\label{K-Lax-a}\\
&\phi_{t_3}=-q_x\phi+(4\eta^2+2q)\phi_x,\label{K-Lax-b}
\end{align}
\ese
where we have taken $\phi_2=\phi$.

\subsubsection{The mKdV family}\label{sec:MKdV}

The mKdV family is derived from the odd numbered equations  of the AKNS hierarchy by taking $(q,r)=(q,\mp q)$, i.e.
\begin{align}\label{M-reduc}
\left(
  \begin{array}{c}
    q \\
    \mp q \\
  \end{array}
\right)_{t_{2s+1}}
=K_{\ss A,2s+1}|_{r=\mp q}=(L_{\ss A}|_{r=\mp q})^{2s}
\left(
  \begin{array}{c}
    q_x \\
    \mp q_x \\
  \end{array}
\right)=
\left(
  \begin{array}{c}
    (L_{\ss M}^\pm)^sq_x \\
    \mp(L_{\ss M}^\pm)^sq_x \\
  \end{array}
\right),\quad
s=0,1,\cdots,
\end{align}
where
\begin{align}\label{M-ro}
L_{\ss M}^\pm=\partial_x^2\pm4q^2\pm4q_x\partial_x^{-1}q,
\end{align}
which leads to the mKdV family
\begin{align}\label{M-hie}
q_{t_{2s+1}}=K_{\ss M,2s+1}^\pm=(L_{\ss M}^\pm)^sq_x,\quad s=0,1,\cdots.
\end{align}
When $s=1$, we obtain the mKdV equation
\begin{align}\label{M-eq}
q_{t_3}=K_{\ss M,3}^\pm=q_{xxx}\pm 6q^2q_x.
\end{align}
The Lax pairs for the mKdV hierarchy are obtained from the Lax pairs for the odd numbered equations of the AKNS hierarchy with the constraint
$(q,r)=(q,\mp q)$.

\subsubsection{The NLS family}\label{sec:NLS}

Replacing $t_s$ with $i^{s-1}t_s$ and taking $r=\mp q^*$ in the AKNS hierarchy \eqref{A-hie}, we have
\begin{align}\label{N-reduc}
\left(
  \begin{array}{c}
    q \\
    \mp q^* \\
  \end{array}
\right)_{t_s}
=K_{\ss A,s}|_{r=\mp q^*}=(-i)^{s-1}(L_{\ss A}|_{r=\mp q^*})^s
\left(
  \begin{array}{c}
    q \\
    \pm q^* \\
  \end{array}
\right),
\end{align}
i.e.
\begin{align}\label{N-hie}
\left(
  \begin{array}{c}
    q \\
    \mp q^* \\
  \end{array}
\right)_{t_s}
=K_{\ss N,s}^\pm=(-i)^{s-1}(L_{\ss N}^\pm)^s
\left(
  \begin{array}{c}
    q \\
    \pm q^* \\
  \end{array}
\right),\quad s=0,1,\cdots,
\end{align}
where
\begin{align}\label{N-ro}
L_{\ss N}^\pm=-\sigma\partial_x+2\sigma
\left(
  \begin{array}{c}
    q \\
    \mp q^* \\
  \end{array}
\right)
\partial_x^{-1}
(\mp q^*,q),
\end{align}
and $^*$ stands for the complex conjugate.
\eqref{N-hie} gives  the NLS family, in which the third equation $(s=2)$ is the NLS equation
\begin{align}\label{N-eq}
q_{t_2}=K_{\ss N,2}^\pm=-i(q_{xx}\pm 2|q|^2q).
\end{align}

We note that sometimes the NLS family also means \eqref{N-hie} only with $s=0,2,4,\cdots$.

\section{The sdAKNS system}\label{sec:dA}

\subsection{The sdAKNS hierarchy and Lax pairs}\label{sec:dA-hie}

\subsubsection{The AL hierarchy}\label{sec:3-1-1}

The sdAKNS hierarchy can be derived from the AL hierarchy \cite{ZC10b}.
Let us  start from the AL spectral problem and its  time evolution part,
\bse\label{AL-Lax}
\begin{align}
&E\b\Phi=\b M\Phi,\quad
\b M=
\left(
  \begin{array}{cc}
    \lambda & Q_n \\
    R_n & \frac{1}{\lambda} \\
  \end{array}
\right),\quad
\b U=
\left(
  \begin{array}{c}
    Q_n \\
    R_n \\
  \end{array}
\right),\quad
\b\Phi=
\left(
  \begin{array}{c}
    \b\phi_1(n) \\
    \b\phi_2(n) \\
  \end{array}
\right),\label{AL-Lax-a}\\
&\b\Phi_{\b t_s}=\b N_s\b\Phi,\quad
\b N_s=
\left(
  \begin{array}{cc}
    \b A_s & \b B_s \\
    \b C_s & \b D_s \\
  \end{array}
\right),\quad s\in \mathbb{Z},\label{AL-Lax-b}
\end{align}
\ese
where $E$ is a shift operator defined as $Ef(n)=f(n+1)$, $\lambda$ is the spectral parameter independent of time,
$Q_n=Q(n,t)$ and $R_n=R(n,t)$ are potential functions,
and $\b A_s,\b B_s,\b C_s$ and $\b D_s$ are Laurent polynomials of $\lambda$ depending on $Q_n$, $R_n$ and their shift operators.
From the discrete zero curvature equation
\begin{align}\label{AL-zcr}
\b M_{\b t_s}=(E\b N_s)\b M_s-\b M\b N_s,
\end{align}
we have
\bse
\begin{align}
&\b A_s=-\frac{1}{\lambda}(E-1)^{-1}(R_nE,-Q_n)
\left(
  \begin{array}{c}
    \b B_s \\
    \b C_s \\
  \end{array}
\right)
+\b A_s^{(0)},\label{AL-zcr-1}\\
&\b D_s=\frac{1}{\lambda}(E-1)^{-1}(R_n,-Q_nE)
\left(
  \begin{array}{c}
    \b B_s \\
    \b C_s \\
  \end{array}
\right)
+\b D_s^{(0)},\label{AL-zcr-4}
\end{align}
\ese
and
\begin{align}\label{AL-zcr-23}
\left(
  \begin{array}{c}
    Q_n \\
    R_n \\
  \end{array}
\right)_{\b t_s}
=(\lambda\b L_1-\frac{1}{\lambda}\b L_2)
\left(
  \begin{array}{c}
    \b B_s \\
    \b C_s \\
  \end{array}
\right)
+(\b A_s^{(0)}-\b D_s^{(0)})
\left(
  \begin{array}{c}
    Q_n \\
    -R_n \\
  \end{array}
\right),
\end{align}
where $\b A_s^{(0)}=\b A_s|_{\b U=0}$, $\b D_s^{(0)}=\b D_s|_{\b U=0}$ and
\bse\label{AL-Lj}
\begin{align}
&\b L_1=
\left(
  \begin{array}{cc}
    -1 & 0 \\
    0 & E \\
  \end{array}
\right)+
\left(
  \begin{array}{c}
    -Q_n \\
    R_nE \\
  \end{array}
\right)
(E-1)^{-1}(R_n,-Q_nE),\label{AL-L1}\\
&\b L_2=
\left(
  \begin{array}{cc}
    -E & 0 \\
    0 & 1 \\
  \end{array}
\right)-
\left(
  \begin{array}{c}
    -Q_nE \\
    R_n \\
  \end{array}
\right)
(E-1)^{-1}(R_nE,-Q_n).\label{AL-L2}
\end{align}
\ese
The inverse of $\b L_1$ and $\b L_2$ can explicitly be written  as
\bse\label{AL-Lj(-1)}
\begin{align}
&\b L_1^{-1}=
\left(
  \begin{array}{cc}
    -1 & 0 \\
    0 & E^{-1} \\
  \end{array}
\right)+
\left(
  \begin{array}{c}
    Q_n \\
    R_{n-1} \\
  \end{array}
\right)
(E-1)^{-1}(R_n,Q_n)\frac{1}{\mu_n},\label{AL-L1(-1)}\\
&\b L_2^{-1}=
\left(
  \begin{array}{cc}
    -E^{-1} & 0 \\
    0 & 1 \\
  \end{array}
\right)-
\left(
  \begin{array}{c}
    Q_{n-1} \\
    R_n \\
  \end{array}
\right)
(E-1)^{-1}(R_n,Q_n)\frac{1}{\mu_n},\label{AL-L2(-1)}
\end{align}
\ese
where
\begin{equation}
\mu_n=1-Q_nR_n.
\label{mu}
\end{equation}

To derive the AL hierarchy, we expand $(\b B_s,\b C_s)^\mathrm{T}$ as
\begin{align}\label{AL-expan-pos}
\left(
  \begin{array}{c}
    \b B_s \\
    \b C_s \\
  \end{array}
\right)
=\sum_{j=1}^s
\left(
  \begin{array}{c}
    \b b_s^{(j)} \\
    \b c_s^{(j)} \\
  \end{array}
\right)
\lambda^{2(s-j)+1},\quad s=1,2,\cdots.
\end{align}
Inserting it into \eqref{AL-zcr-23} and taking $\b A_s^{(0)}=-\b D_s^{(0)}=\frac{1}{2}\lambda^{2s}$, we obtain the following recursion structure
\bse
\begin{align}
& \b U_{t_s}=
-\b L_2 \left(
  \begin{array}{c}
    \b b_s^{(s)} \\
    \b c_s^{(s)} \\
  \end{array}
\right),\\
& \left(
  \begin{array}{c}
    \b b_s^{(j+1)} \\
    \b c_s^{(j+1)} \\
  \end{array}
\right)=
\b L_1^{-1}\b L_2
 \left(
  \begin{array}{c}
    \b b_s^{(j)} \\
    \b c_s^{(j)} \\
  \end{array}
\right),\quad j=1,2,\cdots, s-1,\\
&  \left(
  \begin{array}{c}
    \b b_s^{(1)} \\
    \b c_s^{(1)} \\
  \end{array}
\right)=\b L_1^{-1}
 \left(
  \begin{array}{c}
    -Q_{n} \\
    R_{n} \\
  \end{array}
\right)=
\left(
  \begin{array}{c}
    Q_{n} \\
    R_{n-1} \\
  \end{array}
\right),
\end{align}
\ese
i.e.
\bse
\begin{align}
& \b U_{t_s}=\b K_s=\b L^{s-1}
\left(
  \begin{array}{c}
    Q_{n+1} \\
    -R_{n-1} \\
  \end{array}
\right),\quad s=1,2,\cdots,
\label{AL-hie-p}\\
& \left(
  \begin{array}{c}
    \b b_s^{(j)} \\
    \b c_s^{(j)} \\
  \end{array}
\right)=
\b L_1^{-1} \b L ^{j-1}
\left(
  \begin{array}{c}
    Q_{n} \\
   - R_{n} \\
  \end{array}
\right),\quad  j=1,2,\cdots,s,\label{AL-BC-p}
\end{align}
\ese
where \eqref{AL-hie-p} is the positive AL hierarchy and $\b L$ is its recursion operator defined as
\begin{align}\label{AL-L}
\b L=\b L_2\b L_1^{-1}=&
\left(
  \begin{array}{cc}
    E & 0 \\
    0 & E^{-1} \\
  \end{array}
\right)+
\left(
  \begin{array}{c}
    -Q_nE \\
    R_n \\
  \end{array}
\right)
(E-1)^{-1}(R_nE,Q_nE^{-1})\nonumber\\
&+\mu_n
\left(
  \begin{array}{c}
    -EQ_n \\
    R_{n-1} \\
  \end{array}
\right)
(E-1)^{-1}(R_n,Q_n)\frac{1}{\mu_n}.
\end{align}

Since in the AL spectral problem \eqref{A-sp-a} $\lambda$ and $1/\lambda$ appear symmetrically,
one can also expand $(\b B_s,\b C_s)$ as
\begin{align}\label{AL-expan-neg}
\left(
  \begin{array}{c}
    \b B_s \\
    \b C_s \\
  \end{array}
\right)
=\sum_{j=s}^{-1}
\left(
  \begin{array}{c}
    \b b_s^{(j)} \\
    \b c_s^{(j)} \\
  \end{array}
\right)
\lambda^{2(s-j)-1},\quad  s=-1,-2,\cdots
\end{align}
and take $\b A_s^{(0)}=-\b D_s^{(0)}=\frac{1}{2}\lambda^{2s}$.
In that case,  we obtain
\bse
\begin{align}
& \b U_{t_s}=\b K_s=\b L^{s+1}
\left(
  \begin{array}{c}
    Q_{n-1} \\
    -R_{n+1} \\
  \end{array}
\right),\quad  s = -1,-2,\cdots,
\label{AL-hie-n}\\
& \left(
  \begin{array}{c}
    \b b_s^{(j)} \\
    \b c_s^{(j)} \\
  \end{array}
\right)=
\b L_2^{-1} \b L^{j+1}
\left(
  \begin{array}{c}
    Q_{n} \\
    R_{n-1} \\
  \end{array}
\right),\quad  j=-1,-2,\cdots,s,\label{AL-BC-n}
\end{align}
\ese
where
\begin{align}\label{AL-L(-1)}
\b L^{-1}=\b L_1\b L_2^{-1}={}&
\left(
  \begin{array}{cc}
    E^{-1} & 0 \\
    0 & E \\
  \end{array}
\right)+
\left(
  \begin{array}{c}
    Q_n \\
    -R_nE \\
  \end{array}
\right)
(E-1)^{-1}(R_nE^{-1},Q_nE)\nonumber\\
&+\mu_n
\left(
  \begin{array}{c}
    Q_{n-1} \\
    -ER_n \\
  \end{array}
\right)
(E-1)^{-1}(R_n,Q_n)\frac{1}{\mu_n}.
\end{align}

The positive hierarchy \eqref{AL-hie-p} and negative hierarchy \eqref{AL-hie-n} can be combined together by
defining $\b K_0=(Q_n,-R_n)^\mathrm{T}$,
and the whole AL hierarchy is expressed as (see \cite{ZC02b})
\begin{align}\label{AL-hie}
\b U_{\b t_s}=\b K_s=\b L^s
\left(
  \begin{array}{c}
    Q_n \\
    -R_n \\
  \end{array}
\right),\quad
s\in \mathbb{Z},
\end{align}
Let us display the first few flows below,
\bse\label{AL-Ks}
\begin{align}
&\b K_0=
\left(
  \begin{array}{c}
    Q_n \\
    -R_n \\
  \end{array}
\right),\label{AL-K(0)}\\
&\b K_1=\mu_n
\left(
  \begin{array}{c}
    Q_{n+1} \\
    -R_{n-1} \\
  \end{array}
\right),\label{AL-K(+1)}\\
&\b K_{-1}=\mu_n
\left(
  \begin{array}{c}
    Q_{n-1} \\
    -R_{n+1} \\
  \end{array}
\right),\label{AL-K(-1)}\\
&\b K_2=\mu_n
\left(
  \begin{array}{c}
    \mu_{n+1}Q_{n+2}-Q_{n+1}(Q_{n}R_{n-1}+Q_{n+1}R_{n}) \\
    -\mu_{n-1}Q_{n-2}+R_{n-1}(Q_{n+1}R_{n}+Q_{n}R_{n-1}) \\
  \end{array}
\right),\label{AL-K(+2)}\\
&\b K_{-2}=\mu_n
\left(
  \begin{array}{c}
    \mu_{n-1}Q_{n-2}-Q_{n-1}(Q_{n}R_{n+1}+Q_{n-1}R_{n}) \\
    -\mu_{n+1}Q_{n+2}+R_{n+1}(Q_{n-1}R_{n}+Q_{n}R_{n+1}) \\
  \end{array}
\right).\label{AL-K(-2)}
\end{align}
\ese

As for the Lax pairs, we can first recover $(\b B_s,\b C_s)^{T}$ from the expansions \eqref{AL-expan-pos} and \eqref{AL-expan-neg}
with \eqref{AL-BC-p} and \eqref{AL-BC-n}, respectively,
and then recover $(\b A_s,\b D_s)^T$ from \eqref{AL-zcr-1} and \eqref{AL-zcr-4}.
We list several matrices $\b N_s$ in Appendix \ref{app-A} where $\b N_0$ is for the equation $\b U_{\b t_0}=\b K_0$.

\subsubsection{The sdAKNS hierarchy derived from the AL hierarchy \cite{ZC10b}}

The sdAKNS hierarchy can be given through combining the AL flows in a suitable way.
Define initial flows
\bse\label{dA-aux}
\begin{align}
&\b K_{\ss A,0}=\b K_0,\label{dA-aux-a}\\
&\b K_{\ss A,1}=\frac{1}{2}(\b K_1-\b K_{-1}),\label{dA-aux-b}
\end{align}
\ese
where $\b K_0$ and $\b K_{\pm 1}$ are the AL flows that we obtained previously.
The sdAKNS hierarchy is \cite{ZC10b}
\begin{align}\label{dA-hie}
\b U_{\b t_s}=\b K_{\ss A,s}=
\left\{
\begin{array}{ll}
  \b{\c L}_{\ss A}^j\b K_{\ss A,0}, & s=2j, \\
  \b{\c L}_{\ss A}^j\b K_{\ss A,1}, & s=2j+1
\end{array}
\right.
\end{align}
for $j=0,1,\cdots$, where
\begin{align}\label{dA-ro-II}
\b{\c L}_{\ss A}={}&\b L-2I+\b L^{-1}\nonumber\\
={}&
\left(
  \begin{array}{cc}
    E-2+E^{-1} & 0 \\
    0 & E-2+E^{-1} \\
  \end{array}
\right)+\mu_n
\left(
  \begin{array}{c}
    Q_{n-1}-Q_{n+1}E \\
    R_{n-1}-R_{n+1}E \\
  \end{array}
\right)
(E-1)^{-1}(R_n,Q_n)\frac{1}{\mu_n}\nonumber\\
&+
\left(
  \begin{array}{c}
    -Q_nE \\
    R_n \\
  \end{array}
\right)
(E-1)^{-1}(R_nE,Q_nE^{-1})-
\left(
  \begin{array}{c}
    -Q_n \\
    R_nE \\
  \end{array}
\right)
(E-1)^{-1}(R_nE^{-1},Q_nE).
\end{align}
The first few flows of the sdAKNS hierarchy are
\bse\label{dA-Ks}
\begin{align}
&\b K_{\ss A,0}=
\left(
  \begin{array}{c}
    Q_n \\
    -R_n \\
  \end{array}
\right),\label{dA-K0}\\
&\b K_{\ss A,1}=\frac{1}{2}\mu_n
\left(
  \begin{array}{c}
    Q_{n+1}-Q_{n-1} \\
    R_{n+1}-R_{n-1} \\
  \end{array}
\right),\label{dA-K1}\\
&\b K_{\ss A,2}=
\left(
  \begin{array}{c}
    Q_{n+1}-2Q_n+Q_{n-1}-Q_nR_n(Q_{n+1}+Q_{n-1})\\
    -R_{n+1}+2R_n-R_{n-1}+Q_nR_n(R_{n+1}+R_{n-1}) \\
  \end{array}
\right),\label{dA-K2}\\
&\b K_{\ss A,3}=\frac{1}{2}\mu_n
\left(
  \begin{array}{c}
    \begin{array}{l}
      (E-E^{-1})(Q_{n+1}-2Q_n+Q_{n-1}) \\
      \qquad+Q_{n+1}Q_{n+2}R_{n+1}-Q_n(Q_{n+1}R_{n-1}-Q_{n-1}R_{n+1}) \\
      \qquad-Q_{n-2}Q_{n-1}R_{n-1}-R_n(Q_{n+1}^2-Q_{n-1}^2)
    \end{array}\\
    \begin{array}{l}
      (E-E^{-1})(R_{n+1}-2R_n+R_{n-1}) \\
      \qquad+Q_{n-1}R_{n-2}R_{n-1}+R_n(Q_{n+1}R_{n-1}-Q_{n-1}R_{n+1}) \\
      \qquad-Q_{n+1}R_{n+1}R_{n+2}-Q_n(R_{n+1}^2-R_{n-1}^2)
    \end{array}\\
  \end{array}
\right).\label{dA-K3}
\end{align}
\ese

\subsubsection{The sdAKNS hierarchy and Lax pairs: revisit}\label{sec:dA-Lax}

For the equation
\begin{equation}
\b U_{\b t_s}=\b K_{\ss A,s}
\label{dA-eq-s}
\end{equation}
in the sdAKNS hierarchy, the flow $\b K_{\ss A,s}$ is actually some combination of the AL flows $\{\b K_j\}$. So, we can
suppose that the Lax pair of the equation \eqref{dA-eq-s} is of the following form
\bse\label{dA-Lax}
\begin{align}
&E\b\Phi=\b M_{\ss A}\b\Phi,\quad
\b M_{\ss A}=
\left(
  \begin{array}{cc}
    \lambda & Q_n \\
    R_n & \frac{1}{\lambda} \\
  \end{array}
\right),\label{dA-Lax-a}\\
&\b\Phi_{\b t_s}=\b N_{\ss A,s}\b\Phi,\quad
\b N_{\ss A,s}=
\left(
  \begin{array}{cc}
    \b A_{\ss A,s} & \b B_{\ss A,s} \\
    \b C_{\ss A,s} & \b D_{\ss A,s} \\
  \end{array}
\right),
\label{dA-Lax-b}
\end{align}
\ese
where \eqref{dA-Lax-a} is just  the AL spectral problem \eqref{AL-Lax-a},
and $\b N_{\ss A,s}$ should be the corresponding combinations of $\{\b N_j\}$.
For example, since $\b K_{\ss A, 1}=\frac{1}{2}(\b K_1-\b K_{-1})$, we have
\[\b N_{\ss A,1}=\frac{1}{2}(\b N_{1}-\b N_{-1})
=\frac{1}{2}
\left(
  \begin{array}{cc}
    \frac{1}{2}\lambda^2-Q_nR_{n-1}-\frac{1}{2}\lambda^{-2} & Q_n\lambda+Q_{n-1}\lambda^{-1} \\
    R_{n-1}\lambda+R_n\lambda^{-1} & -\frac{1}{2}\lambda^2-Q_{n-1}R_n+\frac{1}{2}\lambda^{-2} \\
  \end{array}
\right),\]
where $\b N_{\pm 1}$ are given in Appendix \ref{app-A}.
Hoverer, it is hard, for a generic $s$, to give a clear description  that how $\b A_{\ss A,s}$, $\b B_{\ss A,s}$,
$\b C_{\ss A,s}$    and $\b D_{\ss A,s}$ are expressed through the Laurent polynomials in $\lambda$.
Note that the compatible condition (zero curvature equation) of \eqref{dA-Lax} reads
\begin{align}\label{dA-zcr}
\b M_{\ss A,\b t_s}=(E\b N_{\ss A,s})\b M_{\ss A}-\b M_{\ss A}\b N_{\ss A,s},
\end{align}
i.e.
\bse\label{dA-zcr-14}
\begin{align}
&\b A_{\ss A,s}=-\frac{1}{\lambda}(E-1)^{-1}(R_nE,-Q_n)
\left(
  \begin{array}{c}
    \b B_{\ss A,s} \\
    \b C_{\ss A,s} \\
  \end{array}
\right)
+\b A_{\ss A,s}^{(0)},\label{dA-zcr-1}\\
&\b D_{\ss A,s}=\frac{1}{\lambda}(E-1)^{-1}(R_n,-Q_nE)
\left(
  \begin{array}{c}
    \b B_{\ss A,s} \\
    \b C_{\ss A,s} \\
  \end{array}
\right)
+\b D_{\ss A,s}^{(0)}\label{dA-zcr-4}
\end{align}
\ese
and
\begin{align}\label{dA-zcr-23}
\b U_{\b t_s}=(\lambda\b L_1-\frac{1}{\lambda}\b L_2)
\left(
  \begin{array}{c}
    \b B_{\ss A,s} \\
    \b C_{\ss A,s} \\
  \end{array}
\right)
+(\b A_{\ss A,s}^{(0)}-\b D_{\ss A,s}^{(0)})
\left(
  \begin{array}{c}
    Q_n \\
    -R_n \\
  \end{array}
\right),
\end{align}
where $\b A_{\ss A,s}^{(0)}=\b A_{\ss A,s}|_{\b U=0}$ and $\b D_{\ss A,s}^{(0)}=\b D_{\ss A,s}|_{\b U=0}$.

To find a suitable expression of  $\b N_{\ss A,s}$ available for investigating continuous limits,
let us re-derive the sdAKNS hierarchy.
In the following we make use of the G\^ateaux derivative to derive the sdAKNS flows $\{\b K_{\ss A,s}\}$ and the related
$\{\b N_{\ss A,s}\}$.
For the given functions $\b F=\b F(\b U)$ and $\b G=\b G(\b U)$,
\[\b F'[G]=\frac{\partial}{\partial\epsilon} \b F(\b U+\epsilon \b G)\big|_{\epsilon=0}\]
is called the G\^ateaux derivative of $\b F(\b U)$ w.r.t. $\b U$ in the direction $\b G(\b U)$.
For $\b F(\b U)$ it is easy to see that $\b F_{\b t}(\b U)=\b F'[\b U_{\b t}]$.
Adopting this fact, we may rewrite the zero curvature equation \eqref{dA-zcr} as
\begin{align}\label{dA-zcr-gd}
\b M'_{\ss A}[\b K_{\ss A,s}]=(E\b N_{\ss A,s})\b M_{\ss A}-\b M_{\ss A}\b N_{\ss A,s},
\end{align}

To derive the Lax pais of the sdAKNS hierarchy, let us consider
\begin{align}\label{dA-ir}
\b M_{\ss A}'\Big[\b X_{\ss A}-\Big(\lambda-\frac{1}{\lambda}\Big)^2\b Y_{\ss A}\Big]=(E\b{\f N}_{\ss A})\b M_{\ss A}-\b M_{\ss A}\b{\f N}_{\ss A},
\end{align}
where $\b X_{\ss A}=(\b X_{\ss A,1}, \b X_{\ss A,2})^T$ and $\b Y_{\ss A}=(\b Y_{\ss A,1}, \b Y_{\ss A,2})^T$ are
vector functions of $\b U$ but independent of $\lambda$,
and
\begin{align*}
\b{\f N}_{\ss A}=
\left(
  \begin{array}{cc}
    \b{\f A}_{\ss A} & \b{\f B}_{\ss A} \\
    \b{\f C}_{\ss A} & \b{\f D}_{\ss A} \\
  \end{array}
\right).
\end{align*}
When $\b Y_{\ss A}=0$, we assign the following two initial flows
\[\b X_{\ss A}=\b K_{\ss A,0}\quad \hbox{and} \quad  \b X_{\ss A}=\b K_{\ss A,1},\]
where $\b K_{\ss A,0}$ and $\b K_{\ss A,1}$ are defined in \eqref{dA-aux}.
Correspondingly we can take
\[\b {\mathfrak{N}}_{\ss A}=\b N_{\ss A,0}\quad \hbox{and} \quad  \b {\mathfrak{N}}_{\ss A}=\b N_{\ss A,1},\]
respectively.
When $\b Y_{\ss A}\neq 0$, we restrict $\b {\mathfrak{N}}_{\ss A}|_{\b U=0}=0$ and rewrite \eqref{dA-ir} as
\begin{align}\label{dA-ir-23}
\b X_{\ss A}-\Big(\lambda-\frac{1}{\lambda}\Big)^2\b Y_{\ss A}=(\lambda\b L_1-\frac{1}{\lambda}\b L_2)
\left(
  \begin{array}{c}
    \b{\f B}_{\ss A} \\
    \b{\f C}_{\ss A} \\
  \end{array}
\right)
\end{align}
and
\bse\label{dA-zcr-ir-14}
\begin{align}
&\b {\mathfrak{A}}_{\ss A}=-\frac{1}{\lambda}(E-1)^{-1}(R_nE,-Q_n)
\left(
  \begin{array}{c}
    \b {\mathfrak{B}}_{\ss A} \\
    \b {\mathfrak{C}}_{\ss A}
  \end{array}
\right),\label{dA-zcr-ir-1}\\
&\b {\mathfrak{D}}_{\ss A}=\frac{1}{\lambda}(E-1)^{-1}(R_n,-Q_nE)
\left(
  \begin{array}{c}
    \b {\mathfrak{B}}_{\ss A} \\
    \b {\mathfrak{C}}_{\ss A}
  \end{array}
\right).
\label{dA-zcr-ir-4}
\end{align}
\ese
To solve \eqref{dA-ir-23}, we  expand $(\b{\f B}_{\ss A},\b{\f C}_{\ss A})^\mathrm{T}$ as
\begin{align}\label{dA-ir-expan}
\left(
  \begin{array}{c}
    \b{\f B}_{\ss A} \\
    \b{\f C}_{\ss A} \\
  \end{array}
\right)=
\left(
  \begin{array}{c}
    \b{\f b}_{\ss A}^+ \\
    \b{\f c}_{\ss A}^+ \\
  \end{array}
\right)\lambda+
\left(
  \begin{array}{c}
    \b{\f b}_{\ss A}^- \\
    \b{\f c}_{\ss A}^- \\
  \end{array}
\right)\frac{1}{\lambda}
\end{align}
and substitute it into \eqref{dA-ir-23}.
Then we have
\bse\label{dA-ir-coeff}
\begin{align}
&\b Y_{\ss A}=-\b L_1
\left(
  \begin{array}{c}
    \b{\f b}_{\ss A}^+ \\
    \b{\f c}_{\ss A}^+ \\
  \end{array}
\right),\label{dA-ir-coeff-a}\\
&\b X_{\ss A}+2\b Y_{\ss A}=\b L_1
\left(
  \begin{array}{c}
    \b{\f b}_{\ss A}^- \\
    \b{\f c}_{\ss A}^- \\
  \end{array}
\right)-\b L_2
\left(
  \begin{array}{c}
    \b{\f b}_{\ss A}^+ \\
    \b{\f c}_{\ss A}^+ \\
  \end{array}
\right),\label{dA-ir-coeff-b}\\
&\b Y_{\ss A}=\b L_2
\left(
  \begin{array}{c}
    \b{\f b}_{\ss A}^- \\
    \b{\f c}_{\ss A}^- \\
  \end{array}
\right),\label{dA-ir-coeff-c}
\end{align}
\ese
which  gives rise to
\begin{align}\label{dA-ir-XY}
\b X_{\ss A}=\b{\c L}_{\ss A}\b Y_{\ss A}
\end{align}
and
\begin{align}\label{dA-ir-BC}
\left(
  \begin{array}{c}
    \b{\f B}_{\ss A} \\
    \b{\f C}_{\ss A} \\
  \end{array}
\right)
=(-\lambda\b L_1^{-1}+\frac{1}{\lambda}\b L_2^{-1})\b Y_{\ss A},
\end{align}
where $\b{\c L}_{\ss A}=\b L-2I+\b L^{-1}$ is given as \eqref{dA-ro-II}.
This actually indicates the recursive relation
\begin{align}\label{dA-ir-nn}
\b M_{\ss A}'\Big[\b{\c L}_{\ss A}\b Y_{\ss A}-\Big(\lambda-\frac{1}{\lambda}\Big)^2\b Y_{\ss A}\Big]
=(E\b{\f N}_{\ss A})\b M_{\ss A}-\b M_{\ss A}\b{\f N}_{\ss A}.
\end{align}
Repeating such a relation we can reach the form
\begin{align}\label{dA-ir-nnn}
\b M_{\ss A}'\Big[\b{\c L}_{\ss A}^j\b Y_{\ss A}-\Big(\lambda-\frac{1}{\lambda}\Big)^{2j}\b Y_{\ss A}\Big]
=(E\b{\f N}_{\ss A})\b M_{\ss A}-\b M_{\ss A}\b{\f N}_{\ss A},
\end{align}
where in the matrix $\b{\f N}_{\ss A}$, we denote
\begin{align}\label{dA-ir-BC-nnn}
\left(
  \begin{array}{c}
    \b{\f B}_{\ss A} \\
    \b{\f C}_{\ss A} \\
  \end{array}
\right)
=\sum^j_{k=1} \Big(\lambda-\frac{1}{\lambda}\Big)^{2(j-k)}\Bigl(-\lambda\b L_1^{-1}+\frac{1}{\lambda}\b L_2^{-1}\Bigr)\b{\c L}_{\ss A}^{k-1}\b Y_{\ss A}.
\end{align}
Thus, we can take $\b Y_{\ss A}=\b K_{\ss A,0}$ and $\b Y_{\ss A}=\b K_{\ss A,1}$ in \eqref{dA-ir-nnn}, respectively,
and obtain the following zero curvature representation of the flow $\b K_{\ss A,s}$,
\begin{align}\label{dA-zcr-gd-nnn}
\b M'_{\ss A}[\b K_{\ss A,s}]=(E\b N_{\ss A,s})\b M_{\ss A}-\b M_{\ss A}\b N_{\ss A,s}, ~~ s=0,1,\cdots,
\end{align}
where the elements of the matrix $\b N_{\ss A,s}$ are given by
\bse\label{dA-BC}
\begin{align}
\left(
  \begin{array}{c}
    \b B_{\ss A,s} \\
    \b C_{\ss A,s} \\
  \end{array}
\right)
={}&\sum_{k=1}^j\left(\lambda-\frac{1}{\lambda}\right)^{2(j-k)}
\left(-\lambda\b L_1^{-1}+\frac{1}{\lambda}\b L_2^{-1}\right)\b{\c L}_{\ss A}^{k-1}\b K_{\ss A,0},\quad
s=2j,\label{dA-BC-even}\\
\left(
  \begin{array}{c}
    \b B_{\ss A,s} \\
    \b C_{\ss A,s} \\
  \end{array}
\right)
={}&\sum_{k=1}^j\left(\lambda-\frac{1}{\lambda}\right)^{2(j-k)}
\left(-\lambda\b L_1^{-1}+\frac{1}{\lambda}\b L_2^{-1}\right)\b{\c L}_{\ss A}^{k-1}\b K_{\ss A,1}\nonumber\\
&+\frac{1}{2}\left(\lambda-\frac{1}{\lambda}\right)^{2j}
\left(
  \begin{array}{c}
    \lambda Q_n+\frac{1}{\lambda}Q_{n-1} \\
    \lambda R_{n-1}+\frac{1}{\lambda}R_{n} \\
  \end{array}
\right),\quad
s=2j+1\label{dA-BC-odd}
\end{align}
\ese
and
\bse\label{dA-AD}
\begin{align}
&\b A_{\ss A,s}=-\frac{1}{\lambda}(E-1)^{-1}(R_nE,-Q_n)
\left(
  \begin{array}{c}
    \b B_{\ss A,s} \\
    \b C_{\ss A,s} \\
  \end{array}
\right)
+\b A_{\ss A,s}^{(0)},
\label{dA-A}\\
&\b D_{\ss A,s}=\frac{1}{\lambda}(E-1)^{-1}(R_n,-Q_nE)
\left(
  \begin{array}{c}
    \b B_{\ss A,s} \\
    \b C_{\ss A,s}
  \end{array}
\right)+\b D_{\ss A,s}^{(0)},
\label{dA-D}\\
&\b A_{\ss A,s}^{(0)}=-\b D_{\ss A,s}^{(0)}=
\left\{
\begin{array}{ll}
  \frac{1}{2}(\lambda-\frac{1}{\lambda})^s, & s=2j, \\
  \frac{1}{4}(\lambda-\frac{1}{\lambda})^{s-1}(\lambda^2-\frac{1}{\lambda^2}), & s=2j+1,
\end{array}
\right.\label{dA-AD-0}
\end{align}
for $j=0,1,\cdots$.
\ese

The first four $\b N_{\ss A,s}$ are listed in Appendix \ref{app-B}.

We note that the Lax pairs of the sdAKNS hierarchy are unique once we restrict
\[\b N_{\ss A,s}|_{\b U=0}=
\left(\begin{array}{cc} \b A_{\ss A,s}^{(0)}  &  0\\
                        0 & \b D_{\ss A,s}^{(0)}
                        \end{array}
\right),\]
where $\b A_{\ss A,s}^{(0)}$ and $\b D_{\ss A,s}^{(0)}$ are defined as \eqref{dA-AD-0}.
This is guaranteed by the following fact (cf. \cite{ZC02b,ZNBC06}).

\begin{prop}\label{prop:dA-unique}
Suppose that  $\b X=(\b X_{1}, \b X_{2})^T$ is a vector function of $\b U$ but independent of $\lambda$
and $\b N$ is  a $2\times2$ matrix Laurent polynomial in $\lambda$,
living on $\b U$.
Then the matrix equation
\begin{align}\label{dA-test}
\b M_{\ss A}'[\b X]=(E\b N)\b M_{\ss A}-\b M_{\ss A}\b N, \quad \b N|_{\b U=0}=0
\end{align}
has only zero solution $\b X=0,\b N=0$.
\end{prop}

\subsection{Conservation laws}\label{sec:dA-cls}

\subsubsection{Conservation laws: Trivial in continuum limit}\label{sec:dA-cls-tri}

Similar to the continuous case of the  AKNS system, we begin with the spectral problem of the sdAKNS hierarchy, i.e. the AL spectral problem
\bse\label{dA-sp}
\begin{align}
&E\b\phi_1=\lambda\b\phi_1+Q_n\b\phi_2,\label{dA-sp-a}\\
&E\b\phi_2=R_n\b\phi_1+\frac{1}{\lambda}\b\phi_2.\label{dA-sp-b}
\end{align}
\ese
Setting $\b\omega_{\ss A}=\frac{\b\phi_2}{\b\phi_1}$, we arrive at the discrete Riccati equation \cite{ZC02a}
\begin{align}\label{dA-Ric-ome}
\lambda E\b\omega_{\ss A}=\frac{1}{\lambda}\b\omega_{\ss A}-Q_n\b\omega_{\ss A}E\b\omega_{\ss A}+R_n,
\end{align}
which is solved by
\begin{align}\label{dA-ome}
\b\omega_{\ss A}=\sum_{j=1}^{\infty}\b\omega_{\ss A}^{(j)}\lambda^{-2j+1},
\end{align}
with
\bse\label{dA-ome-rec}
\begin{align}
&\b\omega_{\ss A}^{(1)}=R_{n-1},\quad \b\omega_{\ss A}^{(2)}=R_{n-2},\label{dA-ome-rec-a}\\
&\b\omega_{\ss A}^{(j+1)}=E^{-1}\b\omega_{\ss A}^{(j)}
-Q_{n-1}\sum_{k=1}^{j-1}\b\omega_{\ss A}^{(k)}E^{-1}\b\omega_{\ss A}^{(j-k)},\quad j=2,3,\cdots.\label{dA-ome-rec-b}
\end{align}
\ese
From the Lax pair \eqref{dA-Lax} we have
\begin{align*}
(E-1)(\ln\b\phi_1)=\ln (1+\lambda^{-1}Q_n\b\omega_{\ss A}),\quad
(\ln\b\phi_1)_{t_s}=\b A_{\ss A,s}+\b B_{\ss A,s}\,\b\omega_{\ss A},
\end{align*}
which provides the formal conservation law
\begin{align}\label{dA-cf-ome}
\big[\ln(1+\lambda^{-1}Q_n\b\omega_{\ss A})\big]_{\b t_s}=(E-1)(\b A_{\ss A,s}+\b B_{\ss A,s}\, \b\omega_{\ss A}).
\end{align}
Then, for the equation $\b U_{\b t_s}=\b K_{\ss A,s}$ in the sdAKNS hierarchy, with corresponding $\b A_{\ss A,s}$ and $\b B_{\ss A,s}$ in the above formula,
we can expand \eqref{dA-cf-ome} in terms of $\lambda^2$ and get
\begin{equation}
\partial_{\b t_s} \sum^{\infty}_{j=1} \b{\varrho}_{\ss A}^{(j)} \lambda^{-2j}=
(E-1) \sum^{\infty}_{j=1} \b{\mathcal{J}}_{\ss A}^{(j)} \lambda^{-2j}.
\label{dA-cf-ome-2}
\end{equation}
The coefficients of $\lambda^{-2j}$  provide the infinitely many conservation laws for the equation $\b U_{\b t_s}=\b K_{\ss A,s}$:
\begin{equation}
\partial_{\b t_s} \b{\varrho}_{\ss A}^{(j)} =
(E-1)  \b{\mathcal{J}}_{\ss A}^{(j)}£¬\quad  j=1,2,\cdots.
\label{dA-cfs-ome-2}
\end{equation}
However, under the continuum limit given in \cite{ZC10b}, all these conservation laws go to the first conservation law
(i.e. $j=1$ in \eqref{A-cls}) of the continuous case (see Proposition \ref{prop:cl-tri}).
In other words, such infinitely many conservation laws are not new in terms of the continuum limit.

\subsubsection{Conservation laws: Meaningful in continuum limit}\label{sec:dA-cls-mean}

We need to derive new forms of the conservation laws so that they
are meaningful in the continuum limit.
To do so, let us introduce
\begin{align}\label{dA-relat}
\b\Omega_{\ss A}(z)=\frac{1}{\lambda}\b\omega_{\ss A}(\lambda),
\end{align}
where $\lambda$ and $z$ are related through
\begin{align}\label{dA-trans}
\lambda=\sqrt{\frac{1+z}{z}}.
\end{align}

Rewriting the discrete Riccati equation \eqref{dA-Ric-ome} in terms of $\b\Omega_{\ss A}$ and $z$, we obtain
\begin{align}\label{dA-Ric-Ome}
\frac{1}{z}\b\Omega_{\ss A}=(E^{-1}-1)\b\Omega_{\ss A}
-\Big(1+\frac{1}{z}\Big)Q_{n-1}\b\Omega_{\ss A}E^{-1}\b\Omega_{\ss A}+R_{n-1}.
\end{align}
To solve it, the following expansion 
\begin{equation}
\b\Omega_{\ss A}(z)=\sum_{j=1}^{\infty}\b\Omega_{\ss A}^{(j)}z^{j}
\label{dA-Ome-exp}
\end{equation}
yields 
\bse\label{dA-Ome-rec}
\begin{align}
&\b\Omega_{\ss A}^{(1)}=R_{n-1},\quad
\b\Omega_{\ss A}^{(2)}=R_{n-2}(1-Q_{n-1}R_{n-1})-R_{n-1},  \label{dA-Ome-rec-a}\\
&\b\Omega_{\ss A}^{(j+1)}=(E^{-1}-1)\b\Omega_{\ss A}^{(j)}-Q_{n-1}
\sum_{k=1}^{j-1}\b\Omega_{\ss A}^{(k)}E^{-1}\b\Omega_{\ss A}^{(j-k)}-Q_{n-1}
\sum_{k=1}^{j}\b\Omega_{\ss A}^{(k)}E^{-1}\b\Omega_{\ss A}^{(j+1-k)}\label{dA-Ome-rec-b}
\end{align}
\ese
for $j=2,3,\cdots.$
The first few $\Omega_{\ss A}^{(j)}$ are
\bse\label{dA-Omej}
\begin{align}
\b\Omega_{\ss A}^{(1)}={}&R_{n-1},\label{dA-Ome0}\\
\b\Omega_{\ss A}^{(2)}={}&R_{n-2}(1-Q_{n-1}R_{n-1})-R_{n-1},\label{dA-Ome1}\\
\b\Omega_{\ss A}^{(3)}={}&R_{n-1}+2R_{n-2}(Q_{n-1}R_{n-1}-1)+Q_{n-1}R_{n-2}^2(Q_{n-1}R_{n-1}-1)\nonumber\\
&+R_{n-3}(Q_{n-2}R_{n-2}-1)(Q_{n-1}R_{n-1}-1),\label{dA-Ome2}\\
\b\Omega_{\ss A}^{(4)}={}&3R_{n-2}+3Q_{n-1}^2R_{n-2}^2+Q_{n-1}^2R_{n-2}^3-R_{n-1}\nonumber\\
&-3Q_{n-1}R_{n-2}R_{n-1}-3Q_{n-1}^2R_{n-2}^2R_{n-1}-Q_{n-1}^3R_{n-2}^3R_{n-1}\nonumber\\
&-Q_{n-2}R_{n-3}^2(Q_{n-2}R_{n-2}-1)(Q_{n-1}R_{n-1}-1)\nonumber\\
&-R_{n-4}(Q_{n-3}R_{n-3}-1)(Q_{n-2}R_{n-2}-1)(Q_{n-1}R_{n-1}-1)\nonumber\\
&-R_{n-3}(Q_{n-2}R_{n-2}-1)(Q_{n-1}R_{n-2}+3)(Q_{n-1}R_{n-1}-1).\label{dA-Ome3}
\end{align}
\ese
Meanwhile, the formal conservation law \eqref{dA-cf-ome} can be written as
\begin{align}\label{dA-cf-Ome}
\big[\ln(1+Q_n\b\Omega_{\ss A})\big]_{\b t_s}=(E-1)(\b{\c A}_{\ss A,s}+\b{\c B}_{\ss A,s}\b\Omega_{\ss A}),
\end{align}
where
\begin{align}
\b{\c A}_{\ss A,s}(z)=\b A_{\ss A,s}(\lambda)\Bigr|_{\lambda=\sqrt{\frac{1+z}{z}}},\quad
\b{\c B}_{\ss A,s}(z)=\lambda\b B_{\ss A,s}(\lambda)\Bigr|_{\lambda=\sqrt{\frac{1+z}{z}}}.
\label{AB}
\end{align}
Note that we have expansions
\bse
\begin{align}
&\ln(1+Q_n\b\Omega_{\ss A})=\sum_{j=1}^\infty\b\rho_{\ss A}^{(j)}z^j,\label{dA-exp-a}\\
& \b{\c A}_{\ss A,s}+\b{\c B}_{\ss A,s}\b\Omega_{\ss A}=
\b{\mathcal{A}}_{\ss A,s}^{(0)}+\sum_{j=1}^\infty\b J_{\ss A}^{(j)}z^j, \label{dA-exp-b}
\end{align}
where $\b{\mathcal{A}}_{\ss A,s}^{(0)}=\b A_{\ss A,s}^{(0)}|_{\lambda=\sqrt{\frac{1+z}{z}}}$
and $\b A_{\ss A,s}^{(0)}$ is defined by \eqref{dA-AD-0}.
\ese
Then, comparing the coefficients of  $z^j$ of the both sides of \eqref{dA-cf-Ome},
we   obtain infinitely many conservation laws
\begin{align}\label{dA-cls}
\partial_{\b t_s}\b\rho_{\ss A}^{(j)}=(E-1)\b J_{\ss A,s}^{(j)},\quad j=1,2,\cdots.
\end{align}
for the equation $\b U_{\b t_s}=\b K_{\ss A,s}$.
Explicit formulae of $\b\rho_{\ss A}^{(j)}$ can be given with the help of the following proposition (see Proposition 2 in \cite{ZCS-JPA}).
\begin{prop}\label{P:2}
The following expansion holds,
\begin{subequations}\label{exp-t}
\begin{equation}
\ln\biggl(1+\sum_{i=1}^{\infty}y_i z^i
\biggr)=\sum_{j=1}^{\infty}h_j(\bfy)z^{j}, \label{ht}
\end{equation}
where
\begin{equation}
h_j(\bfy)=\sum_{||\balpha||=j}(-1)^{|\balpha|-1}(|\balpha|-1)!\frac{\bfy^{\balpha}}{\balpha
!}, \label{htj}
\end{equation}
and
\begin{align}
& \mathbf{y}=(y_1,y_2,\cdots),\quad \balpha=(\alpha_1,\alpha_2,\cdots),\quad \alpha_i\in\{0,1,\cdots\}, \\
& \bfy^{\balpha}=\prod_{i=1}^{\infty}y_i^{\alpha_i},\quad
 {\balpha}!=\prod_{i=1}^{\infty}(\alpha_i !),\quad
|\balpha|=\sum_{i=1}^{\infty}\alpha_i,\quad
||\balpha||=\sum^{\infty}_{i=1} i\alpha_i.
\end{align}
\end{subequations}
The first few of $\{h_j(\bfy)\}$ are
\begin{subequations}\label{ht1-4}
\begin{align}
& h_1(\bfy)=y_1,\\
& h_2(\bfy)=-\frac{1}{2}y_1^2+y_2, \\
& h_3(\bfy)=\frac{1}{3}y_1^3-y_1y_2+y_3,\\
& h_4(\bfy)=-\frac{1}{4}y_1^4+y_1^2y_2-y_1y_3-\frac{1}{2}y_2^2+y_4.
\end{align}
\end{subequations}
\end{prop}

Thus, for $\b\rho_{\ss A}^{(j)}$ we have
\begin{equation}
\b\rho_{\ss A}^{(j)}=h_j(\bfy),\quad  j=1,2,\cdots,
\label{cl-rhoj-hj}
\end{equation}
with $y_i =Q_n \b\Omega_{\ss A}^{(i)}$.
The first few of $\b\rho_{\ss A}^{(j)}$ are
\bse\label{dA-rhoj}
\begin{align}
\b\rho_{\ss A}^{(1)}={}&Q_nR_{n-1},\label{dA-rho1}\\
\b\rho_{\ss A}^{(2)}={}&-\frac{1}{2}Q_n\big[2R_{n-2}(Q_{n-1}R_{n-1}-1)+R_{n-1}(Q_nR_{n-1}+2)\big],\label{dA-rho2}\\
\b\rho_{\ss A}^{(3)}={}&\frac{1}{3}Q_n^3R_{n-1}^3+Q_n^2R_{n-1}\big[R_{n-1}+R_{n-2}(Q_{n-1}R_{n-1}-1)\big]\nonumber\\
&+Q_n\big[R_{n-1}+2R_{n-2}(Q_{n-1}R_{n-1}-1)+Q_{n-1}R_{n-2}^2(Q_{n-1}R_{n-1}-1)\nonumber\\
&+R_{n-3}(Q_{n-2}R_{n-2}-1)(Q_{n-1}R_{n-1}-1)\big].\label{dA-rho3}
\end{align}
\ese
All the equations in the sdAKNS hierarchy share the same conserved densities.  The associated fluxes depend on the time part of the Lax pairs.
For example, in the case of  $s=2$, we have
\bse\label{dA-cls-AB2}
\begin{align}
&\b{\c A}_{\ss A,2}=\frac{1}{2}\sum_{j=2}^{\infty}(-1)^jz^{-j}-Q_nR_{n-1},\label{dA-A2}\\
&\b{\c B}_{\ss A,2}=Q_n z^{-1}+(Q_n-Q_{n-1}),\label{dA-B2}
\end{align}
\ese
and the first few fluxes are
\bse\label{dA-Jj}
\begin{align}
\b J_{\ss A,2}^{(1)}={}&Q_nR_{n-2}(Q_{n-1}R_{n-1}-1)-Q_{n-1}R_{n-1},\label{dA-J1}\\
\b J_{\ss A,2}^{(2)}={}&Q_n\big[R_{n-3}(Q_{n-2}R_{n-2}-1)+R_{n-2}(Q_{n-1}R_{n-2}-1)\big](Q_{n-1}R_{n-1}-1)\nonumber\\
&+Q_{n-1}\big[R_{n-1}+R_{n-2}(Q_{n-1}R_{n-1}-1)\big],\label{dA-J2}\\
\b J_{\ss A,2}^{(3)}={}&(Q_n-Q_{n-1})\big[R_{n-1}+2R_{n-2}(Q_{n-1}R_{n-1}-1)+Q_{n-1}R_{n-2}^2(Q_{n-1}R_{n-1}-1)\nonumber\\
&+R_{n-3}(Q_{n-2}R_{n-2}-1)(Q_{n-1}R_{n-1}-1)\big]+Q_n\big[3R_{n-2}+3Q_{n-1}R_{n-2}^2\nonumber\\
&+Q_{n-1}^2R_{n-2}^3-R_{n-1}-3Q_{n-1}R_{n-2}R_{n-1}-3Q_{n-1}^2R_{n-2}^2R_{n-1}\nonumber\\
&-Q_{n-1}^3R_{n-2}^3R_{n-1}-Q_{n-2}R_{n-3}^2(Q_{n-2}R_{n-2}-1)(Q_{n-1}R_{n-1}-1)\nonumber\\
&-R_{n-4}(Q_{n-3}R_{n-3}-1)(Q_{n-2}R_{n-2}-1)(Q_{n-1}R_{n-1}-1)\nonumber\\
&-R_{n-3}(Q_{n-2}R_{n-2}-1)(Q_{n-1}R_{n-2}+3)(Q_{n-1}R_{n-1}-1)\big].\label{dA-J3}
\end{align}
\ese

The infinitely many conservation laws are not trivial in the continuum limit (see Sec.\ref{sec:cl-cls}).

\subsubsection{Conservation laws: Combinatorial relation}\label{sec:dA-cls-comb}

The conserved density $\b\rho_{\ss A}^{(j)}$ is actually certain combination of the conserved densities $\{\b{\varrho}_{\ss A}^{(i)}\}$
and the same combinatorial relation is used by those conservation laws with $\b\rho_{\ss A}^{(j)}$.
Let us first prove the following lemma.

\begin{lem}\label{lem-da-cl-com}
If
\begin{equation}
\sum_{j=1}^{\infty}h_j(\bfy)z^j=\sum_{s=1}^{\infty}h_s(\bfx)\biggl(\sum^{\infty}_{k=1}(-1)^{k-1}z^k\biggr)^s,
\label{dA-cf-hy-hx}
\end{equation}
where $\bfx=(x_1,x_2,\cdots)$,
then we have
\begin{equation}
h_j(\bfy)=\sum_{s=1}^{j}(-1)^{j-s}\mathrm{C}^{s-1}_{j-1} h_s(\bfx), \quad j=1,2,\cdots,
\label{dA-cf-hyx-j}
\end{equation}
where
\[\mathrm{C}^n_m=\frac{m!}{n!(m-n)!},\quad m\geq n.\]
\end{lem}

\begin{proof}
Expanding the r.h.s. of \eqref{dA-cf-hy-hx} in terms of $z$ and comparing the coefficients of $z^j$, we can find that
\bse\label{dA-cf-hyx-1-4}
\begin{align}
& h_1(\bfy)=h_1(\bfx),\\
& h_2(\bfy)=h_2(\bfx)-h_1(\bfx),\\
& h_3(\bfy)=h_3(\bfx)-2h_2(\bfx)+h_1(\bfx),\\
& h_4(\bfy)=h_4(\bfx)-3h_3(\bfx)+3h_2(\bfx)-h_1(\bfx),
\end{align}
\ese
which cope with the formula \eqref{dA-cf-hyx-j}.
Let us go to prove that \eqref{dA-cf-hyx-j} holds for generic $j$.
This is equivalent to prove
\begin{equation}
\biggl(\sum^{\infty}_{k=1}(-1)^{k-1}z^k\biggr)^s
=\sum^{\infty}_{j=s}(-1)^{j-s}\mathrm{C}^{s-1}_{j-1} z^j.
\label{dA-cf-hxy-z-sj}
\end{equation}
Based on \eqref{dA-cf-hyx-1-4}, let us suppose that \eqref{dA-cf-hxy-z-sj} is true for $s\leq i$.
Then, when $s=i+1$ we have
\begin{align*}
\biggl(\sum^{\infty}_{k=1}(-1)^{k-1}z^k\biggr)^{i+1}
& =\biggl(\sum^{\infty}_{k=1}(-1)^{k-1}z^k\biggr)\,\biggl(\sum^{\infty}_{k=1}(-1)^{k-1}z^k\biggr)^{i}\\
& =\biggl(\sum^{\infty}_{k=1}(-1)^{k-1}z^k\biggr)\,\sum^{\infty}_{l=i}(-1)^{l-i}\mathrm{C}^{i-1}_{l-1} z^l\\
& = \sum^{\infty}_{j=i+1}(-1)^{j-i-1} \biggl( \sum^{j-1}_{l=i}\mathrm{C}^{i-1}_{l-1} \biggr) z^j.
\end{align*}
Then, by the combinatorial formula
\[\sum^m_{k=0}\mathrm{C}^{n}_{n+k}=\mathrm{C}^{n+1}_{n+m+1},\]
we immediately obtain
\[\biggl(\sum^{\infty}_{k=1}(-1)^{k-1}z^k\biggr)^{i+1}
= \sum^{\infty}_{j=i+1}(-1)^{j-(i+1)}  \mathrm{C}^{i}_{j-1}  \, z^j,\]
which means \eqref{dA-cf-hxy-z-sj} is true for $s=i+1$.
Thus, thanks to the mathematical inductive method, we complete the proof.
\end{proof}

Now, noting that with the help of the polynomials $\{h_j(\bfy)\}$ we have
$\b\rho_{\ss A}^{(j)}=h_j(\bfy)$ and $\b{\varrho}_{\ss A}^{(j)}=h_j(\bfx)$
where $y_i=Q_n \b\Omega_{\ss A}^{(i)}$ and $x_i=Q_n \b\omega_{\ss A}^{(i)}$,
we immediately reach the following relation for $\b\rho_{\ss A}^{(j)}$ and $\b{\varrho}_{\ss A}^{(j)}$.

\begin{prop}
\label{prop:aA-cq-com}
The conserved densities $\b\rho_{\ss A}^{(j)}$ and $\b{\varrho}_{\ss A}^{(j)}$
obey the following combinatorial relation,
\begin{equation}
\b\rho_{\ss A}^{(j)} =\sum_{s=1}^{j}(-1)^{j-s}\mathrm{C}^{s-1}_{j-1} \b{\varrho}_{\ss A}^{(s)}, \quad j=1,2,\cdots.
\label{dA-cf-cq-com-j}
\end{equation}
\end{prop}

\subsection{Reductions}\label{sec:dA-reduc}

\subsubsection{The sdKdV hierarchy}\label{sec:sdKdV}

Considering  the odd-numbered equations in the sdAKNS hierarchy under the constraint $(Q_n,R_n)=(Q_n,-1)$, we have
\begin{align}\label{dK-reduc}
\left(
  \begin{array}{c}
    Q_n \\
    -1 \\
  \end{array}
\right)_{\b t_{2s+1}}
&=\b K_{\ss A,2s+1}|_{R_n=-1}=(\b{\c L}_{\ss A}|_{R_n=-1})^s
\left(
  \begin{array}{c}
    \frac{1}{2}\mu_n(Q_{n+1}-Q_{n-1}) \\
    0 \\
  \end{array}
\right)\nonumber\\
&=
\left(
  \begin{array}{c}
    \b L_{\ss K}^s\frac{\mu_n }{2}(Q_{n+1}-Q_{n-1}) \\
    0 \\
  \end{array}
\right),\quad s=0,1,\cdots,
\end{align}
where $\mu_n=1+Q_n$ and
\begin{align}\label{dK-ro}
\b L_{\ss K}=\mu_n(E+E^{-1})-(2-Q_n-Q_{n+1})+\mu_n(Q_{n+1}-Q_{n-1})(E-1)^{-1}\frac{1}{\mu_n}.
\end{align}
Thus, the sdKdV hierarchy is given by
\begin{align}\label{dK-hie}
Q_{n,\b t_{2s+1}}=\b K_{\ss K,2s+1}=\b L_{\ss K}^s\frac{1}{2}\mu_n(Q_{n+1}-Q_{n-1}),\quad s=0,1,\cdots,
\end{align}
and $\b L_{\ss K}$ is the recursion operator of the sdKdV hierarchy.
The first two equations are
\bse
\label{dK-2-eq}
\begin{align}
& Q_{n,\b t_1}=\b K_{\ss K,1}=\frac{1}{2}\mu_n(Q_{n+1}-Q_{n-1}), \label{dK-volterra}\\
& Q_{n,\b t_3}=\b K_{\ss K,3}=\frac{1}{2}\mu_n(E-E^{-1})\big[Q_{n+1}-2Q_n+Q_{n-1}+Q_n(Q_{n+1}+Q_n+Q_{n-1})\big].\label{dK-eq}
\end{align}
\ese

The Lax pairs of those odd-numbered equations in the sdAKNS hierarchy under the constraint $(Q_n,R_n)=(Q_n,-1)$
are reduced to the Lax pairs of the sdKdV hierarchy.
With regards to the scalar forms, eliminating $\b\phi_1$ from the Lax pairs,
one can obtain the scalar form of Lax pairs for the sdKdV hierarchy:
\bse\label{dK-Lax}
\begin{align}
&E^2\b\phi=\Big(\lambda+\frac{1}{\lambda}\Big)E\b\phi-(1+Q_n)\b\phi,\label{dK-Lax-a}\\
&\b\phi_{\b t_{2j+1}}=\alpha_{2j+1}\b\phi+\beta_{2j+1}E\b\phi,\label{dK-Lax-b}
\end{align}
\ese
where we have taken $\b\phi_2=\b\phi$.
For the two equations in \eqref{dK-2-eq},
we have
\begin{align*}
\alpha_1={}&-\frac{1}{4}\lambda^2+\frac{1}{2}(Q_{n-1}-1)-\frac{1}{4}\lambda^{-2},\quad \beta_1=\frac{1}{2}\lambda+\frac{1}{2}\lambda^{-1},\\
\alpha_3={}&-\frac{1}{4}\lambda^4-\frac{1}{2}Q_n\lambda^2+\frac{1}{2}\big[Q_{n-2}(1+Q_{n-1})\\
&+(Q_{n-1}-1)(Q_{n-1}+Q_n-1)\big]-\frac{1}{2}Q_n\lambda^{-2}-\frac{1}{4}\lambda^{-4},\\
\beta_3={}&\frac{1}{2}\lambda^3+\frac{1}{2}(Q_{n-1}+Q_n-1)\lambda+\frac{1}{2}(Q_{n-1}+Q_n-1)\lambda^{-1}+\frac{1}{2}\lambda^{-3}.
\end{align*}

The sdKdV hierarchy obtained here are related to the Volterra (also known as the Langmuir or Kac-van Moerbeke, cf. \cite{MP96a}) hierarchy.
Starting from the linear problems
\bse\label{V-Lax}
\begin{align}
&E^2\b\phi=\zeta E\b\phi-V_n\b\phi,\label{V-Lax-a}\\
&\b\phi_{\b t_{2s+1}}=\b A_{\ss V,2s+1}\b\phi+\b B_{\ss V,2s+1}E\b\phi,\quad s=0,1,\cdots,\label{V-Lax-b}
\end{align}
\ese
where $V_n=V(n,t)$ is a potential function and $\zeta$ is the spectral parameter,
the Volterra hierarchy can be cast as follows,
\begin{align}\label{V-hie}
V_{n,\b t_{2s+1}}=\b K_{\ss V,2s+1}=\b L_{\ss V}^{s}\frac{1}{2}V_n(E-E^{-1})V_n,\quad s=0,1,\cdots,
\end{align}
where the recursion operator $\b L_{\ss V}$ is
\begin{align}\label{V-ro}
\b L_{\ss V}=V_n(1+E^{-1})(EV_nE-V_n)(E-1)^{-1}\frac{1}{V_n}.
\end{align}
The first two equations in the Volterra hierarchy are
\bse\label{V-eq}
\begin{align}
&\b V_{n,\b t_1}=\b K_{\ss V,1}=\frac{1}{2}V_n(E-E^{-1})V_n,\label{V-eq1}\\
&\b V_{n,\b t_3}=\b K_{\ss V,3}=\frac{1}{2}V_n(E-E^{-1})\big[V_n(V_{n+1}+V_n+V_{n-1})\big]-V_n(E-E^{-1})V_n.\label{V-eq2}
\end{align}
\ese
Both equations $\b V_{n,\b t}=\b K_{\ss V,1}$ and $\b V_{n,\b t}=\b K_{\ss V,3}-4\b K_{\ss V,1}$ can be viewed as the
discretizations of the KdV equation \eqref{K-eq} (cf. \cite{MP96a,HHLRW00,Sur03}).
In fact, the sdKdV hierarchy \eqref{dK-hie} and the Volterra hierarchy \eqref{V-hie} are related  through
\begin{subequations}
\label{VK}
\begin{align}
& V_n=1+Q_n,\label{VK-1}\\
& \b L_{\ss K}=\b L_{\ss V}-4,  \label{ro-VK}
\end{align}
\end{subequations}
which  reveals the following combinatorial relations of the sdKdV flows and the Volterra flows,
\begin{align}\label{dK-combin}
\b K_{\ss K,2s+1}=\sum_{i=0}^s
(-4)^{s-i}\,\mathrm{C}^{i}_{s}\,
\b K_{\ss V,2i+1}\Big|_{V_n=1+Q_n},\quad  s=0,1,\cdots.
\end{align}

\subsubsection{The sdmKdV hierarchy}\label{sec:sdMKdV}

To get the sdmKdV hierarchy, taking $R_n=\mp Q_n$ in the odd-numbered equations in the sdAKNS hierarchy \eqref{dA-hie} yields 
\begin{align}\label{dM-reduc}
\left(
  \begin{array}{c}
    Q_n \\
    \mp Q_n \\
  \end{array}
\right)_{\b t_{2s+1}}
&=\b K_{\ss A,2s+1}|_{R_n=\mp Q_n}=(\b{\c L}_{\ss A}|_{R_n=\mp Q_n})^s
\left(
  \begin{array}{c}
    \frac{1}{2}\mu_n^\pm(Q_{n+1}-Q_{n-1}) \\
    \frac{1}{2}\mu_n^\pm(Q_{n+1}-Q_{n-1}) \\
  \end{array}
\right)\nonumber\\
&=
\left(
  \begin{array}{c}
    (\b L_{\ss M}^\pm)^s\frac{1}{2}\mu_n^\pm(Q_{n+1}-Q_{n-1}) \\
    \mp(\b L_{\ss M}^\pm)^s\frac{1}{2}\mu_n^\pm(Q_{n+1}-Q_{n-1}) \\
  \end{array}
\right),\quad s=0,1,\cdots,
\end{align}
where $\mu_n^\pm=1\pm Q_n^2$ and
\begin{align}\label{dM-ro}
\b L_{\ss M}^\pm=\mu_n^\pm(E+E^{-1})-(2\mp Q_nQ_{n+1})\pm\mu_n^{\pm}(Q_{n+1}-Q_{n-1})(E-1)^{-1}Q_n\frac{1}{\mu_n^\pm}.
\end{align}
Thus, the sdmKdV hierarchy is written as
\begin{align}\label{dM-hie}
Q_{n,\b t_{2s+1}}=\b K_{\ss M,2s+1}^\pm=(\b L_{\ss M}^\pm)^s\frac{1}{2}\mu_n^\pm(Q_{n+1}-Q_{n-1}),\quad s=0,1,\cdots,
\end{align}
where $\b L_{\ss M}^\pm$ is the recursion operator.
The first two equations in the hierarchy are
\bse
\begin{align}
&Q_{n,\b t_1}=\frac{1}{2}\mu_n^\pm(Q_{n+1}-Q_{n-1}),\label{dM-eq-0}\\
&Q_{n,\b t_3}=\b K_{\ss M,3}^\pm=\frac{1}{2}\mu_n^{\pm}(E-E^{-1})\big[Q_{n+1}-2Q_n+Q_{n-1}\pm Q_n^2(Q_{n+1}+Q_{n-1})\big].
\label{dM-eq}
\end{align}
\ese
The Lax pairs for the sdmKdV hierarchy are able to be obtained from  those of the odd-numbered equations in the sdAKNS hierarchy by taking $R_n=\mp Q_n$.

\subsubsection{The sdNLS hierarchy}\label{sec:sdNLS}

Taking $\b t_s=i^{s-1}\b t_s$, then the reduction $(Q_n,R_n)=(Q_n,\mp Q_n^*)$ of the sdAKNS hierarchy \eqref{dA-hie} leads to
\bse\label{dN-reduc}
\begin{align}
\left(
  \begin{array}{c}
    Q_n \\
    \mp Q_n^* \\
  \end{array}
\right)_{\b t_s}
&=(-i)^{s-1}\b K_{\ss A,s}|_{R_n=\mp Q_n^*}=(-i)^{s-1}(\b{\c L}_{\ss A}|_{R_n=\mp Q_n^*})^j
\left(
  \begin{array}{c}
    Q_n \\
    \pm Q_n^* \\
  \end{array}
\right)\nonumber\\
&=(-i)^{s-1}(\b{\c L}_{\ss N}^\pm)^j
\left(
  \begin{array}{c}
    Q_n \\
    \pm Q_n^* \\
  \end{array}
\right),\quad s=2j,\label{dN-reduc-even}
\end{align}
\begin{align}
\left(
  \begin{array}{c}
    Q_n \\
    \mp Q_n^* \\
  \end{array}
\right)_{\b t_s}
&=(-i)^{s-1}\b K_{\ss A,s}|_{R_n=\mp Q_n^*}=(-i)^{s-1}(\b{\c L}_{\ss A}|_{R_n=\mp Q_n^*})^j\frac{1}{2}\mu_n^\pm
\left(
  \begin{array}{c}
    Q_{n+1}-Q_{n-1} \\
    \mp(Q_{n+1}^*-Q_{n-1}^*) \\
  \end{array}
\right)\nonumber\\
&=(-i)^{s-1}(\b{\c L}_{\ss N}^\pm)^j\frac{1}{2}\mu_n^\pm
\left(
  \begin{array}{c}
    Q_{n+1}-Q_{n-1} \\
    \mp(Q_{n+1}^*-Q_{n-1}^*) \\
  \end{array}
\right),\quad s=2j+1\label{dN-reduc-odd}
\end{align}
\ese
for $j=0,1,\cdots$, where $\mu_n^\pm=1\pm|Q_n|^2$ and
\begin{align}\label{dN-ro}
\b{\c L}_{\ss N}^\pm={}&
\left(
  \begin{array}{cc}
    E-2+E^{-1} & 0 \\
    0 & E-2+E^{-1} \\
  \end{array}
\right)\nonumber\\
&+
\left(
  \begin{array}{c}
    -Q_nE \\
    \mp Q_n^* \\
  \end{array}
\right)
(E-1)^{-1}(\mp Q_n^*E,Q_nE^{-1})-
\left(
  \begin{array}{c}
    -Q_n \\
    \mp Q_n^*E \\
  \end{array}
\right)
(E-1)^{-1}(\mp Q_n^*E^{-1},Q_nE)\nonumber\\
&+\mu_n^\pm
\left(
  \begin{array}{c}
    -EQ_n \\
    \mp Q_{n-1}^* \\
  \end{array}
\right)
(E-1)^{-1}(\mp Q_n^*,Q_n)\frac{1}{\mu_n^\pm}+\mu_n^\pm
\left(
  \begin{array}{c}
    Q_{n-1} \\
    \pm EQ_n^* \\
  \end{array}
\right)
(E-1)^{-1}(\mp Q_n^*,Q_n)\frac{1}{\mu_n^\pm}.
\end{align}
The sdNLS hierarchy are then present the following,
\begin{align}\label{dN-hie}
\left(
  \begin{array}{c}
    Q_n \\
    \mp Q_n^* \\
  \end{array}
\right)_{\b t_s}
=\b K_{\ss N,s}^\pm=(-i)^{s-1}(\b{\c L}_{\ss N}^\pm)^j
\left\{
\begin{array}{ll}
\left(
  \begin{array}{c}
    Q_n \\
    \pm Q_n^* \\
  \end{array}
\right), & s=2j, \\
\frac{1}{2}\mu_n^\pm
\left(
  \begin{array}{c}
    Q_{n+1}-Q_{n-1} \\
    \mp(Q_{n+1}^*-Q_{n-1}^*) \\
  \end{array}
\right), & s=2j+1
\end{array}
\right.
\end{align}
for $j=0,1,\cdots$.
The third equation in the hierarchy, i.e.
\begin{align}\label{dN-eq}
Q_{n,\b t_2}=\b K_{\ss N,2}^\pm=-i\big[Q_{n+1}-2Q_n+Q_{n-1}\pm|Q_n|^2(Q_{n+1}+Q_{n-1})\big]
\end{align}
is known as the sdNLS equation \cite{APT-book} (or the AL equation).
We note that sometimes the sdNLS hierarchy also means those equations with only $s=2j$ in \eqref{dN-hie}.

\section{Continuum limits}\label{sec:cl}

Since the integrable discretizations usually break the original dispersion relations,
it is not easy, in general, to give a uniform continuum limit, which maps the discrete integrable systems together with their integrable
characteristics to the continuous counterparts (cf. \cite{MP96a,MP98b,MP98c,Sch82}).
In \cite{ZC10b} we have presented a unform continuum limit which sends the whole sdAKNS hierarchy to the AKNS hierarchy.
This continuum limit  also explained the structure deformation
of the Lie algebra of symmetries.
In the following we use the same continuum limit to investigate the
Lax pairs, conservation laws and reductions of the sdAKNS hierarchy.

\subsection{Plan}\label{sec:cl-plan}

Our plan for the continuum limit runs below \cite{ZC10b}:
\begin{itemize}
\item{Replacing $Q_n$ and $R_n$ with $hq_n$ and $hr_n$, where $h$ is the real spacing parameter.}
\item{Let $n\to\infty$ and $h\to 0$ such that $nh$ finite.}
\item{Define continuous variable $x=x_0+nh$, then for a scalar function, for example, $q_n$, one has $q_{n+j}=q(x+jh)$. For convenience we take $x_0=0$.}
\item{Define time coordinate relation $t_s=h^s\b t_s$ for $s=0,1,\cdots$.}
\item{Continuous spectral parameter $\eta$ is defined by $\lambda=e^{h\eta}$.}
\end{itemize}

\subsection{Hierarchy}\label{sec:cl-hie}

In Ref. \cite{ZC10b} we have shown that in the above continuum limit  the sdAKNS hierarchy \eqref{dA-hie}
goes to the continuous AKNS hierarchy \eqref{A-hie}.
Let us briefly review these results.

In the continuum limit described in Sec.\ref{sec:cl-plan}, it can be shown that
\bse\label{cl-K01}
\begin{align}
&\b K_{\ss A,0}=K_{\ss A,0}h+O(h^2),\label{cl-K0}\\
&\b K_{\ss A,1}=K_{\ss A,1}h^2+O(h^3),\label{cl-K1}
\end{align}
\ese
and
\begin{align}\label{cl-ro-II}
\b{\c L}_{\ss A}=\b L_{\ss A}^2h^2+O(h^3).
\end{align}
Then, from the recursive structure of the sdAKNS hierarchy \eqref{dA-hie}, the continuum limits for the flows are
\begin{align}\label{cl-flow}
\b K_{\ss A,s}=K_{\ss A,s}h^{s+1}+O(h^{s+2}),\quad s=0,1,\cdots,
\end{align}
and at the level of equations  we have the following.
\begin{prop}\label{prop-cl-hie}
In the continuum limit described in Sec.\ref{sec:cl-plan}, we have
\begin{align}\label{cl-hie}
\b U_{\b t_s}-\b K_{\ss A,s}=(u_{t_s}-K_{\ss A,s})h^{s+1}+O(h^{s+2}),\quad s=0,1,\cdots.
\end{align}
\end{prop}

\subsection{Lax pairs}\label{sec:cl-Lax}

Based on the continuum limit designed in Sec.\ref{sec:cl-plan}, it is easy for one to find that the continuum limit of the spectral problem \eqref{dA-Lax-a}
is
\begin{align}\label{cl-Lax-a}
E\b\Phi-\b M_{\ss A}\b\Phi=(\Phi_x-M_{\ss A}\Phi)h+O(h^2).
\end{align}
To investigate the relations in the continuum limit between the time parts of the Lax pairs \eqref{A-Lax} and \eqref{dA-Lax},
we rewrite \eqref{A-BC} as the following form:
\bse\label{A-BC-oe}
\begin{align}
&\left(
  \begin{array}{c}
    B_{\ss A,s} \\
    C_{\ss A,s} \\
  \end{array}
\right)
=-\sigma\sum_{k=1}^{j}(2\eta)^{2(j-k)}(L_{\ss A}+2\eta)L_{\ss A}^{2(k-1)}
\left(
  \begin{array}{c}
    q \\
    -r \\
  \end{array}
\right),\quad s=2j,\label{A-BC-even}\\
&\left(
  \begin{array}{c}
    B_{\ss A,s} \\
    C_{\ss A,s} \\
  \end{array}
\right)
=-\sigma\sum_{k=1}^{j}(2\eta)^{2(j-k)}(L_{\ss A}+2\eta)L_{\ss A}^{2k-1}
\left(
  \begin{array}{c}
    q \\
    -r \\
  \end{array}
\right)
+(2\eta)^{2j}
\left(
  \begin{array}{c}
    q \\
    r \\
  \end{array}
\right),\quad s=2j+1\label{A-BC-odd}
\end{align}
\ese
for $j=0,1,\cdots$.
A direct calculation yields
\begin{align}\label{cl-T}
\lambda-\frac{1}{\lambda}=2\eta h+O(h^2),\quad
-\lambda\b L_1^{-1}+\frac{1}{\lambda}\b L_2^{-1}=(L_{\ss A}+2\eta)h+O(h^2).
\end{align}
Then, from the expression \eqref{dA-BC} we have
\begin{align*}
\left(
  \begin{array}{c}
    \b B_{\ss A,s} \\
    \b C_{\ss A,s} \\
  \end{array}
\right)
={}&\Bigg[
-\sigma\sum_{k=1}^{j}(2\eta)^{2(j-k)}(L_{\ss A}+2\eta)L_{\ss A}^{2(k-1)}
\left(
  \begin{array}{c}
    q \\
    -r \\
  \end{array}
\right)
\Bigg]
h^{2j}+O(h^{2j+1}),\quad s=2j,\\
\left(
  \begin{array}{c}
    \b B_{\ss A,s} \\
    \b C_{\ss A,s} \\
  \end{array}
\right)
={}&\Bigg[-\sigma\sum_{k=1}^{j}(2\eta)^{2(j-k)}(L_{\ss A}+2\eta)L_{\ss A}^{2k-1}
\left(
  \begin{array}{c}
    q \\
    -r \\
  \end{array}
\right)\\
&+(2\eta)^{2j}
\left(
  \begin{array}{c}
    q \\
    r \\
  \end{array}
\right)
\Bigg]h^{2j+1}+O(h^{2j+2}),\quad s=2j+1
\end{align*}
for $j=0,1,\cdots$, namely
\begin{align}\label{cl-BC}
\left(
  \begin{array}{c}
    \b B_{\ss A,s} \\
    \b C_{\ss A,s} \\
  \end{array}
\right)
=
\left(
  \begin{array}{c}
    B_{\ss A,s} \\
    C_{\ss A,s} \\
  \end{array}
\right)
h^{s}+O(h^{s+1}),\quad s=0,1,\cdots.
\end{align}
Next, substituting \eqref{cl-BC} into \eqref{dA-AD} we   find
\bse\label{cl-AD}
\begin{align}
&\b A_{\ss A,s}=A_{\ss A,s}h^s+O(h^{s+1}),\label{cl-A}\\
&\b D_{\ss A,s}=D_{\ss A,s}h^s+O(h^{s+1}), \label{cl-D}
\end{align}
\ese
for $s=0,1,\cdots$.
Therefore, from the relations \eqref{cl-BC} and \eqref{cl-AD} we conclude that
\begin{align}\label{cl-N}
\b N_{\ss A,s}=N_{\ss A,s}h^s+O(h^{s+1}),\quad s=0,1,\cdots,
\end{align}
and further, we have the following.
\begin{prop}\label{prop-cl-lax}
In the continuum limit described in Sec.\ref{sec:cl-plan}, we have
\bse
\begin{align}\label{cl-Lax-b}
&E\b\Phi-\b M_{\ss A}\b\Phi=(\Phi_x-M_{\ss A}\Phi)h+O(h^2),\\
&\b\Phi_{\b t_s}-\b N_{\ss A,s}\b\Phi=(\Phi_{t_s}-N_{\ss A,s}\Phi)h^s+O(h^{s+1}),\quad s=0,1,\cdots.
\end{align}
\ese
\end{prop}

\subsection{Conservation laws}\label{sec:cl-cls}

In this part we will investigate the continuum limit of the infinitely many conservation laws obtained in Sec.\ref{sec:dA-cls}
for the sdAKNS hierarchy \eqref{dA-hie}.
We need to examine the relations between $\b\Omega_{\ss A}^{(j)}$ and $\omega_{\ss A}^{(j)}$,
$\b\Omega_{\ss A}$ and $\omega_{\ss A}$,
the Riccati equations \eqref{dA-Ric-Ome} and \eqref{A-Ric},
the formal conservation laws \eqref{dA-cf-Ome} and \eqref{A-cf},
and the infinitely many conservation laws \eqref{dA-cls} and \eqref{A-cls}, respectively.

First, for the relation  between $\b\Omega_{\ss A}^{(j)}$ and $\omega_{\ss A}^{(j)}$, we have the following result.
\begin{lem}
In the continuum limit described in Sec.\ref{sec:cl-plan}, we have
\begin{align}\label{cl-Omej}
\b\Omega_{\ss A}^{(j)}=\omega_{\ss A}^{(j)}h^{j}+O(h^{j+1}),\quad j=1,2,\cdots,
\end{align}
where  $\b\Omega_{\ss A}$ and $\omega_{\ss A}^{(j)}$ are defined in \eqref{dA-Ome-rec} and \eqref{A-ome-rec}, respectively.
\end{lem}
\begin{proof}
Let us use mathematical induction method. For $j=1,2$, we can find that
\begin{align*}
&\b\Omega_{\ss A}^{(1)}=R_{n-1}=rh+O(h^2),\\
&\b\Omega_{\ss A}^{(2)}=R_{n-2}(1-Q_{n-1}R_{n-1})-R_{n-1}=r_x h^2 +O(h^3),
\end{align*}
which means the relation \eqref{cl-Omej} holds for $j=1,2$.
Now, we suppose
\begin{align*}
\b\Omega_{\ss A}^{(i)}=\omega_{\ss A}^{(i)}h^{i}+O(h^{i+1})
\end{align*}
are true for any $i\leq j$.
Then, from the recursive relation \eqref{dA-Ome-rec-b} and after some  calculation we find
\begin{align*}
\b\Omega_{\ss A}^{(j+1)}&=(E^{-1}-1)\b\Omega_{\ss A}^{(j)}-Q_{n-1}\sum_{k=1}^{j-1}\b\Omega_{\ss A}^{(k)}E^{-1}\b\Omega_{\ss A}^{(j-k)}
-Q_{n-1}\sum_{k=1}^{j}\b\Omega_{\ss A}^{(k)}E^{-1}\b\Omega_{\ss A}^{(j+1-k)}\\
&=\Big(-\omega_{\ss A,x}^{(j)}-q\sum_{k=1}^{j-1}\omega_{\ss A}^{(k)}\omega_{\ss A}^{(j-k)}\Big)h^{(j+1)}+O(h^{j+2})\\
&=\omega_{\ss A}^{(j+1)}h^{(j+1)}+O(h^{j+2}),
\end{align*}
where the last equality coincides with the recursive relation \eqref{A-ome-rec-b}.
Therefore  \eqref{cl-Omej} holds for any $j\geq 1$.
\end{proof}

Let us now go back to the expansion \eqref{dA-Ome-exp}, i.e.
\begin{equation}
\b\Omega_{\ss A}(z)=\sum_{j=1}^{\infty}\b\Omega_{\ss A}^{(j)}z^{j}.
\label{dA-Ome-exp-2}
\end{equation}
Noting that
\begin{align}\label{cl-z}
z=\frac{1}{\lambda^2-1}=(2\eta)^{-1}h^{-1}+O(1),
\end{align}
and inserting \eqref{cl-Omej} and \eqref{cl-z} into \eqref{dA-Ome-exp-2} we have
\begin{equation}
\b\Omega_{\ss A}(z)=\sum_{j=1}^{\infty}\b\Omega_{\ss A}^{(j)}z^{j}=\sum_{j=1}^{\infty} \omega_{\ss A}^{(j)}(2\eta)^{-j}+O(h),
\label{dA-Ome-ome-exp}
\end{equation}
namely,
\begin{align}\label{cl-Ome}
\b\Omega_{\ss A}=\omega_{\ss A}+O(h).
\end{align}
This gives the relation between $\b\Omega_{\ss A}$ and $\omega_{\ss A}$.

Next, substituting the above relation into the Riccati equation \eqref{dA-Ric-Ome} leads to
\begin{align}\label{cl-Ric}
&\frac{1}{z}\b\Omega_{\ss A}-\Big[(E^{-1}-1)\b\Omega_{\ss A}
-\Big(1+\frac{1}{z}\Big)Q_{n-1}\b\Omega_{\ss A}E^{-1}\b\Omega_{\ss A}+R_{n-1}\Big]\nonumber\\
={}&\big[2\eta\omega_{\ss A}-(\omega_{\ss A,x}-q\omega_{\ss A}^2+r)\big]h+O(h^2),
\end{align}
which means the Riccati equation \eqref{dA-Ric-Ome} goes to the continuous Riccati equation \eqref{A-Ric} in the continuum limit.

Let us look at the formal conservation law \eqref{dA-cf-Ome}.
For the l.h.s. of \eqref{dA-cf-Ome}, by using \eqref{cl-Ome} it is easy for us to see that
\begin{align}\label{cl-cd}
\ln(1+Q_n\b\Omega_{\ss A})=(q\omega_{\ss A})h+O(h^2).
\end{align}
Meanwhile, from the relations \eqref{AB} and \eqref{cl-N} we immediately reach
\bse\label{cl-cAB}
\begin{align}
&\b{\c A}_{\ss A,s}=A_{\ss A}h^s+O(h^{s+1}),\label{cl-cA}\\
&\b{\c B}_{\ss A,s}=B_{\ss A}h^s+O(h^{s+1}),\label{cl-cB}
\end{align}
\ese
which provides
\begin{align}\label{cl-flux}
\b{\c A}_{\ss A,s}+\b{\c B}_{\ss A,s}\b\Omega_{\ss A}=(A_{\ss A,s}+B_{\ss A,s}\omega_{\ss A})h^s+O(h^{s+1}).
\end{align}
Thus, for the relation of the  formal conservation law of \eqref{dA-cf-Ome} and \eqref{A-cf}, we have the following result.
\begin{lem}
In the continuum limit described in Sec.\ref{sec:cl-plan}, we have
\begin{align}
&\big[\ln(1+Q_n\b\Omega_{\ss A}\big]_{\b t_s}-(E-1)(\b{\c A}_{\ss A,s}+\b{\c B}_{\ss A,s}\b\Omega_{\ss A})\nonumber\\
=&{}\big[(q\omega_{\ss A})_{t_s}-(A_{\ss A,s}+B_{\ss A,s}\omega_{\ss A})_x\big]h^{s+1}+O(h^{s+2}),
\end{align}
which describes the relation of the two formal conservation laws.
\end{lem}

Finally, we focus on the continuum limits of the explicit infinitely many conservation laws \eqref{dA-cls},
i.e.,
\begin{align}\label{dA-cls-cl}
\partial_{\b t_s}\b\rho_{\ss A}^{(j)}=(E-1)\b J_{\ss A,s}^{(j)},\quad j=1,2,\cdots.
\end{align}
The common  conserved densities $\b\rho_{\ss A}^{(j)}$ are determined by
\begin{align*}
&\ln(1+Q_n\b\Omega_{\ss A})=\sum_{j=1}^\infty\b\rho_{\ss A}^{(j)}z^j,
\end{align*}
with the explicit formulae \eqref{cl-rhoj-hj}, i.e.
\begin{equation}
\b\rho_{\ss A}^{(j)}=h_j(\bfy),\quad  j=1,2,\cdots,
\label{cl-rhoj-hj-new}
\end{equation}
with $y_i\doteq  Q_n \b\Omega_{\ss A}^{(i)}$.

To investigate the continuum limit of $\b\rho_{\ss A}^{(j)}$, we introduce \textit{degrees} of functions (cf. \cite{ZC10b}).
\begin{defn}
Under the plan described in Sec.\ref{sec:cl-plan}, a function $\b F(\b U)$  can be expanded as
a series  of $h$.
The order of the leading term of the series is called the \textit{degree} of $\b F(\b U)$,
denoted by $\deg\b F$.
\end{defn}

For example,
\[\deg  Q_n=1,\quad  \deg \b\Omega_{\ss A}^{(i)}=i, \quad \deg y_i=i+1,\quad \deg \b t_s=s, \quad \deg z=-1.\]
Now, looking at the definition of the polynomials $h_j(\bfy)$ \eqref{htj} (as examples see \eqref{ht1-4}), since
$\deg y_i=i+1$, we can find that $\deg h_j(\bfy)=\deg y_j=j+1$ and
\[\lim_{h\to 0}\frac{h_j(\bfy)}{h^{j+1}}=\lim_{h\to 0} \frac{y_j}{h^{j+1}}.\]
Since $y_j\doteq  Q_n \b\Omega_{\ss A}^{(j)}=q \omega_{\ss A}^{(j)} h^{j+1}+O(h^{j+2})$, we immediately have
\begin{equation}
\b\rho_{\ss A}^{(j)}=h_j(\bfy)=q \omega_{\ss A}^{(j)} h^{j+1}+O(h^{j+2}),
\label{cl-rhoj-cl}
\end{equation}
which means in the continuum limit we have $\b\rho_{\ss A}^{(j)}\to \rho_{\ss A}^{(j)}=q \omega_{\ss A}^{(j)}$.

Next, noting that $\deg \b\rho_{\ss A}^{(j)}=j+1$, which is given by \eqref{cl-rhoj-cl}, then it follows from \eqref{dA-cls-cl} that
$\deg \b J_{\ss A,s}^{(j)}$ must be $j+s$.
Therefore, we can suppose that
\begin{equation}
\b J_{\ss A,s}^{(j)}=W_{\ss A,s}^{(j)} h^{j+s}+O(h^{j+s+1}).
\end{equation}
This means, in light of \eqref{cl-z}, we have
\begin{equation}
\b J_{\ss A,s}^{(j)}z^j=W_{\ss A,s}^{(j)} (2\eta)^{-j}h^s +O(h^{s+1}),
\end{equation}
and further, from \eqref{dA-exp-b} we get
\[\b{\c A}_{\ss A,s}+\b{\c B}_{\ss A,s}\b\Omega_{\ss A}-\b{\c A}_{\ss A,s}^{(0)}
=\sum_{j=1}^\infty\b J_{\ss A,s}^{(j)}z^j
=\biggl(\sum_{j=1}^\infty W_{\ss A,s}^{(j)} (2\eta)^{-j}\biggr)h^s +O(h^{s+1}).\]
Meanwhile, from \eqref{cl-flux} and \eqref{A-expn-flux} we find
\begin{align*}
\b{\c A}_{\ss A,s}+\b{\c B}_{\ss A,s}\b\Omega_{\ss A}-\b{\c A}_{\ss A,s}^{(0)}
=&\big[A_{\ss A,s}+B_{\ss A,s}\omega_{\ss A}-\frac{1}{2}(2\eta)^s\big]h^s+O(h^{s+1})\\
=& \biggl(\sum_{j=1}^\infty J_{\ss A,s}^{(j)} (2\eta)^{-j}\biggr)h^s +O(h^{s+1}).
\end{align*}
Thus, comparing the term of $(2\eta)^{-j}$ immediately yields $W_{\ss A,s}^{(j)}=J_{\ss A,s}^{(j)}$.
It gives rise to
\begin{equation}
\b J_{\ss A,s}^{(j)}=J_{\ss A,s}^{(j)} h^{j+s}+O(h^{j+s+1}).
\label{cl-Jj-cl}
\end{equation}

Thus, for the infinitely many conservation laws of the sdAKNS hierarchy, we can conclude in the following proposition.
\begin{prop}\label{prop:cl-cls}
In the continuum limit described in Sec.\ref{sec:cl-plan}, we have
\begin{subequations}\label{cl-cls-cqf}
\begin{align}
& \b\rho_{\ss A}^{(j)}= \rho_{\ss A}^{(j)} h^{j+1}+O(h^{j+2}),\\
& \b J_{\ss A,s}^{(j)}=J_{\ss A,s}^{(j)} h^{j+s}+O(h^{j+s+1}),
\end{align}
\end{subequations}
and
\begin{align}
\partial_{\b t_s} \b\rho_{\ss A}^{(j)} -(E-1)\b J_{\ss A,s}^{(j)}
=(\partial_{t_s}\rho_{\ss A}^{(j)}-\partial_{x} J_{\ss A,s}^{(j)}) h^{j+s+1}+ O(h^{j+s+2}),
\label{cl-cls-cls}
\end{align}
for $s=0,1,\cdots$ and $j=1,2,\cdots$,
which describes the relation of the two sets of infinitely many conservation laws.
\end{prop}

At the end of this part, let us also take a look at the continuum limit of the conserved densities $\b\varrho_{\ss A}^{(j)}=h_j(\bfx)$ with
$x_i\doteq  Q_n \b\omega_{\ss A}^{(i)}$, which are derived in Sec.\ref{sec:dA-cls-tri} (also see \cite{ZC02a}).
Obviously,
\[
\b\omega_{\ss A}^{(1)}=R_{n-1}=h r+O(h^{2}),\quad  \b\omega_{\ss A}^{(2)}=R_{n-2}=h r+O(h^{2}),
\]
which means $\deg \b\omega_{\ss A}^{(1)}=\deg \b\omega_{\ss A}^{(2)}=1$.
Then, from the recursive structure \eqref{dA-ome-rec} we immediately have
\[\deg \b\omega_{\ss A}^{(j)} \equiv 1,\quad  j=1,2,\cdots,\]
and
\[
\b\omega_{\ss A}^{(j)}= h r+O(h^{2}),\quad  j=1,2,\cdots.
\]
Thus we get
\[x_j=qr h^2+O(h^3),\quad j=1,2,\cdots,\]
and then from the definition of $h_j(\bfx)$ we obtain 
\[h_j(\bfx)\equiv qr h^2 +O(h^{3}),\quad  j=1,2,\cdots.\]

\begin{prop}\label{prop:cl-tri}
All the  conserved densities $\b\varrho_{\ss A}^{(j)}=h_j(\bfx)$ with
$x_i\doteq  Q_n \b\omega_{\ss A}^{(i)}$, which are derived in Sec.\ref{sec:dA-cls-tri} (also see \cite{ZC02a})
for the AL hierarchy as well as for the sdAKNS hierarchy,
are trivial in the continuum limit given in Sec.\ref{sec:cl-plan},
and
\[ \b\varrho_{\ss A}^{(j)}\equiv  qr h^2 +O(h^{3}),\quad  j=1,2,\cdots.\]
In other words, in our continuum limit, all of the conserved densities $\b\varrho_{\ss A}^{(j)}$ go to
$\rho_{\ss A}^{(1)}=qr$, which is the first conserved density of the AKNS hierarchy.
\end{prop}

\subsection{Reductions}\label{sec:cl-reduc}

\subsubsection{The sdKdV hierarchy}

For the sdKdV hierarchy,
if we still use the continuum limit scheme given in Sec.\ref{sec:cl-plan}, we need to
rewrite the reduction of the sdAKNS hierarchy under the constraint  $(Q_n,R_n)=(h q_n,-h)$.
However, this is just equivalent to taking $Q_n=h^2 q_n$  if we still use the constraint  $(Q_n,R_n)=(Q_n,-1)$.
Thus, for the continuum limit of the sdKdV hierarchy, we still follow the scheme proposed in Sec.\ref{sec:cl-plan}
except that the first item is replaced by
\begin{itemize}
\item{ Replacing $Q_n$ with $h^2 q_n$.}
\end{itemize}
Then we find that in the continuum limit
\begin{align}\label{cl-K-ro}
\b L_{\ss K}=L_{\ss K}h^2+O(h^3),
\end{align}
and
\[
Q_{n,\b t_{1}}-\b K_{\ss K,1}=(q_{t_{1}}-q_x)h^{3}+O(h^{4}),
\]
which implies 
\begin{align}\label{cl-K-hie}
Q_{n,\b t_{2s+1}}-\b K_{\ss K,2s+1}=(q_{t_{2s+1}}-K_{\ss K,2s+1})h^{2s+3}+O(h^{2s+4}),\quad s=0,1,\cdots.
\end{align}

We note that in our continuum limit it is the second equation \eqref{dK-eq} in the sdKdV hierarchy \eqref{dK-hie}
that goes to the continuous KdV equation, namely
\begin{align}\label{cl-K-eq}
&Q_{n,t_3}-\frac{1}{2}\mu_n(E-E^{-1})\big[Q_{n+1}-2Q_n+Q_{n-1}+Q_n(Q_{n+1}+Q_n+Q_{n-1})\big]\nonumber\\
={}&(q_{t_3}-q_{xxx}-6qq_x)h^5+O(h^6).
\end{align}
 This is also found in \cite{HHLRW00,Sur03}.
In this sense, equation \eqref{dK-eq} can be referred to as the sdKdV equation.
Meanwhile, for the Lax pair of \eqref{dK-eq}, we find
\bse\label{cl-K-Lax}
\begin{align}
&E^2\b\phi-\Big(\lambda+\frac{1}{\lambda}\Big)E\b\phi-(1+Q_n)\b\phi=\big[\phi_{xx}-(\eta^2-q)\phi\big]h^2+O(h^3),\label{cl-K-Lax-a}\\
&\b\phi_{\b t_3}-\alpha_3 \b\phi-\beta_3 E\b\phi=\big[\phi_{t_3}+q_x\phi-(4\eta^2+2q)\phi_x\big]h^3+O(h^4),\label{cl-K-Lax-b}
\end{align}
\ese
i.e. it goes to the Lax pair of the KdV equation.

\subsubsection{The sdmKdV hierarchy}

For the sdmKdV hierarchy,
under the continuum limit scheme described  in Sec.\ref{sec:cl-plan},
we find
\begin{align}\label{cl-M-ro}
\b L_{\ss M}^\pm=L_{\ss M}^\pm h^2+O(h^3)
\end{align}
and
\[
Q_{n,\b t_{1}}-\b K_{\ss M,1}^{\pm}=(q_{t_{1}}-q_x)h^{2}+O(h^{3}),
\]
which  leads to
\begin{align}\label{cl-M-hie}
Q_{n,\b t_{2s+1}}-\b K_{\ss M,2s+1}^{\pm}=(q_{t_{2s+1}}-K_{\ss M,2s+1}^{\pm})h^{2s+2}+O(h^{2s+3}),\quad s=0,1,\cdots.
\end{align}

We call equation \eqref{dM-eq} the sdmKdV equation because in our continuum limit,
it goes to the mKdV equation (see also \cite{AL76,Sur03}), i.e.
\begin{align}\label{cl-M-eq}
&Q_{n,\b t_3}-\frac{1}{2}\mu_n^\pm(E-E^{-1})\big[Q_{n+1}-2Q_n+Q_{n-1}\pm Q_n^2(Q_{n+1}+Q_{n-1})\big]\nonumber\\
={}&\big[q_{t_3}-(q_{xxx}\pm 6q^2q_x)\big]h^4+O(h^5).
\end{align}
A higher order equation $\b Q_{n,\b t_5}=\b K_{\ss M,5}^-$ was investigated in a recent paper \cite{ZZH13}.

\subsubsection{The sdNLS hierarchy}

For the sdNLS hierarchy, we find that in the continuum limit,
\begin{align}\label{cl-N-ro}
\b{\c L}_{\ss N}^\pm=(L_{\ss N}^\pm)^2 h^2+O(h^3),
\end{align}
and
\begin{align*}
& \left(
  \begin{array}{c}
    Q_n \\
    \mp Q_n^* \\
  \end{array}
\right)_{\b t_0}
-\b K_{\ss N,0}^{\pm}=\Bigg[
\left(
  \begin{array}{c}
    q \\
    \pm q^* \\
  \end{array}
\right)_{t_0}
-i\left(
  \begin{array}{c}
    q \\
    \mp q^* \\
  \end{array}
\right)
\Bigg]h^{}+O(h^{2})\\
& \left(
  \begin{array}{c}
    Q_n \\
    \mp Q_n^* \\
  \end{array}
\right)_{\b t_1}
-\b K_{\ss N,1}^{\pm}=\Bigg[
\left(
  \begin{array}{c}
    q \\
    \mp q^* \\
  \end{array}
\right)_{t_1}
-\left(
  \begin{array}{c}
    q \\
    \mp q^* \\
  \end{array}
\right)_x
\Bigg]h^{2}+O(h^{3}).
\end{align*}
Then we have
\begin{align}\label{cl-N-hie}
\left(
  \begin{array}{c}
    Q_n \\
    \mp Q_n^* \\
  \end{array}
\right)_{\b t_s}
-\b K_{\ss N,s}^{\pm}=\Bigg[
\left(
  \begin{array}{c}
    q \\
    \mp q^* \\
  \end{array}
\right)_{  t_s}
-K_{\ss N,s}^{\pm}
\Bigg]h^{s+1}+O(h^{s+2}),\quad s=0,1,\cdots.
\end{align}
For example, the continuum limit of the sdNLS equation \eqref{dN-eq} is
\begin{align}
&Q_{n,\b t_2}-i\big[Q_{n+1}-2Q_n+Q_{n-1}\pm|Q_n|^2(Q_{n+1}+Q_{n-1})\big]\nonumber\\
={}&\big[q_{t_2}-i(q_{xx}\pm 2|q|^2q)\big]h^3+O(h^4).
\end{align}

\section{Conclusions}\label{sec:concl}

We have shown that the sdAKNS hierarchy can be derived from the AL hierarchy through certain combinations.
As reductions we obtained the sdKdV, sdmKdV and sdNLS hierarchies.
The Lax pairs and conservation laws of the sdAKNS hierarchy were re-derived
so that they cope with their continuous counterparts in  continuum limit.
We designed a uniform continuum limit scheme under which the sdAKNS hierarchy,
their Lax pairs, infinitely many conservation laws and reductions of the hierarchy
go to their continuous counterparts of the AKNS system.
In this continuum limit scheme the spatial and temporal independence is kept.
The same scheme has been also used to explain the structure deformation of symmetry algebra \cite{ZC10b}.
Hamiltonian structures of the sdAKNS hierarchy and their continuum limit will be investigated later.

Finally, a further comment is given for conservation laws.
The known infinitely many conservation laws derived in \cite{ZC02a} for the AL hierarchy
(as well as for the sdAKNS hierarchy) are trivial in light of the continuum limit.
We have shown that all those conserved densities $\b\varrho_{\ss{A}}^{(j)}$ go to the same $\rho_{\ss{A}}^{(1)}=qr$
which is the first conserved density of the AKNS hierarchy.
We have given  \textit{new} forms of the conservation laws for the sdAKNS hierarchy.
These new conservations laws cope with their continuous counterparts and
are related to those trivial ones through explicit combinatorial relation.
By the same constraints we used in reductions, we can obtain the infinitely many conservation laws of those reduced hierarchies.

\section*{Acknowledgements}

The authors thank Prof. Deng-yuan Chen for enthusiastic discussions.
DJZ is grateful to Prof. Zhijun Qiao for the hospitality when he visited
the University of Texas-Pan American.
This project is  partially supported by National Natural Science Foundation of China (Grant Numbers 11071157 and 11171295),
the SRF of the DPHE of China (No. 20113108110002)
and the Project of ``First-class Discipline of Universities in Shanghai''.

\appendix

\section{$\b N_s$ in the AL Lax pair \eqref{AL-Lax}}\label{app-A}

{We list out several matrices $\b N_s$ in the AL Lax pair \eqref{AL-Lax}. $\b N_0$ is for the equation $\b u_{\b t_0}=\b K_0$.}
\bse\label{AL-Ns}
\begin{align}
&\b N_0=
\left(
  \begin{array}{cc}
    \frac{1}{2} & 0 \\
    0 & -\frac{1}{2} \\
  \end{array}
\right),\label{AL-N(0)}\\
&\b N_1=
\left(
  \begin{array}{cc}
    \frac{1}{2}\lambda^2-Q_nR_{n-1} & Q_n\lambda \\
    R_{n-1}\lambda & -\frac{1}{2}\lambda^2 \\
  \end{array}
\right),\quad \b N_{-1}=
\left(
  \begin{array}{cc}
    -\frac{1}{2}\lambda^{-2} & -Q_{n-1}\lambda^{-1} \\
    -R_n\lambda^{-1} & -\frac{1}{2}\lambda^{-2}+Q_{n-1}R_n \\
  \end{array}
\right),\label{AL-N(-1)}
\end{align}
\begin{align}
&\b N_2=
\left(
  \begin{array}{cc}
    \b A_2 & \b B_2 \\
    \b C_2 & \b D_2 \\
  \end{array}
\right),\quad \b N_2=
\left(
  \begin{array}{cc}
    \b A_{-2} & \b B_{-2} \\
    \b C_{-2} & \b D_{-2} \\
  \end{array}
\right),\label{AL-N(-2)}
\end{align}
\ese
where
\begin{align*}
&\b A_2=\frac{1}{2}\lambda^4-Q_nR_{n-1}\lambda^2-\mu_{n-1}Q_nR_{n-2}-\mu_nQ_{n+1}R_{n-1}+Q_{n-1}+Q_n^2R_{n-1}^2,\\
&\b B_2=Q_n\lambda^3+(\mu_nQ_{n+1}-Q_n^2R_{n-1})\lambda,\quad
\b C_2=R_{n-1}\lambda^3+(\mu_{n-1}R_{n-2}-Q_nR_{n-1}^2)\lambda,\\
&\b D_2=-\frac{1}{2}\lambda^4+Q_nR_{n-1}\lambda^2,
\end{align*}
and
\begin{align*}
&\b A_{-2}=\frac{1}{2}\lambda^{-4}-Q_{n-1}R_n\lambda^{-2},\\
&\b B_{-2}=-Q_{n-1}\lambda^{-3}+(-\mu_{n-1}Q_{n-2}+Q_{n-1}^2R_n)\lambda^{-1},\\
&\b C_{-2}=-R_{n-1}\lambda^{-3}+(-\mu_nR_{n+1}+Q_{n-1}R_n^2)\lambda^{-1},\\
&\b D_{-2}=-\frac{1}{2}\lambda^{-4}+Q_{n-1}R_n\lambda^{-2}+\mu_{n-1}Q_{n-2}R_n+\mu_{n}Q_{n-1}R_{n+1}-Q_{n-1}^2R_n^2.
\end{align*}

\section{The first four $\b N_{\ss A,s}$ }\label{app-B}

Here we give the first four of $\b N_{\ss A,s}$:
\bse\label{dA-Ns}
\begin{align}
&\b N_{\ss A,0}=
\left(
  \begin{array}{cc}
    \frac{1}{2} & 0 \\
    0 & -\frac{1}{2} \\
  \end{array}
\right),\label{dA-N0}\\
&\b N_{\ss A,1}=\frac{1}{2}
\left(
  \begin{array}{cc}
    \frac{1}{2}\lambda^2-Q_nR_{n-1}-\frac{1}{2}\lambda^{-2} & Q_n\lambda+Q_{n-1}\lambda^{-1} \\
    R_{n-1}\lambda+R_n\lambda^{-1} & -\frac{1}{2}\lambda^2-Q_{n-1}R_n+\frac{1}{2}\lambda^{-2} \\
  \end{array}
\right),\label{dA-N1}\\
&\b N_{\ss A,2}=
\left(
  \begin{array}{cc}
    \frac{1}{2}\lambda^2-(1+Q_nR_{n-1})+\frac{1}{2}\lambda^{-2} & Q_n\lambda-Q_{n-1}\lambda^{-1} \\
    R_{n-1}\lambda-R_n\lambda^{-1} & -\frac{1}{2}\lambda^2+(1+Q_{n-1}R_n)-\frac{1}{2}\lambda^{-2} \\
  \end{array}
\right),\label{dA-N2}\\
&\b N_{\ss A,3}=
\left(
  \begin{array}{cc}
    \b A_{\ss A,3} & \b B_{\ss A,3} \\
    \b C_{\ss A,3} & \b D_{\ss A,3} \\
  \end{array}
\right),\label{dA-N3}
\end{align}
\ese
where
\begin{align*}
\b A_{\ss A,3}={}&\frac{1}{4}\lambda^4-\frac{1}{2}(1+Q_nR_{n-1})\lambda^2-\frac{1}{2}(Q_nR_{n-2}+Q_{n+1}R_{n-1}-2Q_nR_{n-1}\\
&-Q_{n-1}Q_nR_{n-2}R_{n-1}-Q_nQ_{n+1}R_{n-1}R_n-Q_n^2R_{n-1}^2)+\frac{1}{2}(1+Q_{n-1}R_n)\lambda^{-2}-\frac{1}{4}\lambda^{-4},\\
\b B_{\ss A,3}={}&\frac{1}{2}Q_n\lambda^3+\frac{1}{2}(Q_{n+1}-2Q_n-Q_nQ_{n+1}R_n-Q_n^2R_{n-1})\lambda\\
&+\frac{1}{2}(Q_{n-2}-2Q_{n-1}-Q_{n-2}Q_{n-1}R_{n-1}-Q_{n-1}^2R_n)\lambda^{-1}+\frac{1}{2}Q_{n-1}\lambda^{-3},\\
\b C_{\ss A,3}={}&\frac{1}{2}R_{n-1}\lambda^3+\frac{1}{2}(R_{n-2}-2R_{n-1}-Q_{n-1}R_{n-2}R_{n-1}-Q_nR_{n-1}^2)\lambda\\
&+\frac{1}{2}(R_{n+1}-2R_n-Q_nR_nR_{n+1}-Q_{n-1}R_n^2)\lambda^{-1}+\frac{1}{2}R_n\lambda^{-3},\\
\b D_{\ss A,3}={}&-\frac{1}{4}\lambda^4+\frac{1}{2}(1+Q_nR_{n-1})\lambda^2-\frac{1}{2}(Q_{n-2}R_n+Q_{n-1}R_{n+1}-2Q_{n-1}R_n\\
&-Q_{n-2}Q_{n-1}R_{n-1}R_n-Q_{n-1}Q_nR_nR_{n+1}-Q_{n-1}^2R_n^2)-\frac{1}{2}(1+Q_{n-1}R_n)\lambda^{-2}+\frac{1}{4}\lambda^{-4}.
\end{align*}

\end{document}